\documentclass[12pt,reqno]{amsart}

\usepackage{amscd,amssymb,amsthm}
\usepackage{graphicx}
\usepackage{multirow}
\usepackage{array}
\usepackage{color}
\usepackage{bbm}
\usepackage{epstopdf}

\newtheorem{theorem}{Theorem}[section]
\newtheorem{corollary}{Corollary}[section]
\newtheorem{proposition}{Proposition}[section]
\newtheorem{lemma}{Lemma}[section]

\newtheorem{definition}{Definition}[section]

\numberwithin{equation}{section}

\def\asubs#1{{
\vbox{\hrule height .2pt \kern 1pt \hbox{$\scriptstyle {#1}\kern
1pt$} }\kern-.05pt \vrule width .2pt }}
\def\hang{\hangindent=\parindent\noindent}
\def\urltilda{\kern -.15em\lower .7ex\hbox{\~{}}\kern .04em}

\begin{document}
\def\hang{\hangindent=\parindent\noindent}

\thispagestyle{empty}

\vspace*{1mm}

\begin{center}
\MakeUppercase{\bf {Multiple risk factor dependence structures: Distributional properties}}

\bigskip

\noindent
{\large Jianxi Su, Edward Furman$^{*}$}

\noindent
Department of Mathematics and Statistics, York University, Toronto, Ontario M3J 1P3, Canada

\bigskip
\end{center}

\begin{quote}
\textbf{Abstract.}
We introduce a class of dependence structures, that we call the Multiple Risk Factor (MRF) dependence structures. On the one hand, the new constructions extend the popular CreditRisk$^+$ approach, and as such
they formally describe default risk portfolios exposed to an arbitrary number of fatal risk factors with conditionally
exponential and dependent hitting (or occurrence) times. On the other hand, the MRF structures can be seen as an
encompassing family of multivariate probability distributions with univariate margins distributed Pareto of the 2nd kind, and in
this role they can be used to model insurance risk portfolios of dependent and heavy tailed risk components.

\bigskip \noindent\textit{Keywords and phrases: }
Multivariate distributions, dependence, Pareto distributions, default risk, factor
models, weighted risk measures.


\end{quote}

\newpage

\begin{section}{Introduction}

Consider an $n(\in\mathbf{N})$-variate random vector $\mathbf{X}:=(X_1,\ldots,X_n)'$ with coordinates
$X_i:=E_{\lambda_i}\wedge E_{\lambda_{n+1}},\ i=1,\ldots,n$, where $E_{\lambda_j}$ are exponentially
distributed and stochastically independent random variables (r.v.'s) with parameters $\lambda_j\in\mathbf{R}_+,\ j=1,\ldots,n+1$. The joint distribution of $\mathbf{X}$ is nowadays well-known as the Marshall \& Olkin (MO) multivariate
exponential distribution (Marshall \& Olkin (1967)).

An interesting peculiarity of the MO distribution is that it is not
absolutely continuous with respect to the corresponding Lebesgue measure. While this feature is rather an impediment when
it comes to univariate cumulative distribution functions (c.d.f.'s), the singularity is often perceived as an advantage in
multivariate extensions. Specifically, actuaries have found the MO multivariate exponential distribution quite useful in the
context of life insurance, where the coordinates $X_1,\ldots,X_n$ are being interpreted as dependent future life times of $n$
insureds with $E_{\lambda_i},\ i=1,\ldots,n$  and $E_{\lambda_{n+1}}$ denoting, respectively, the future life times due to
the individual mortality and the  `hitting' time of a common shock risk factor (r.f.).

Speaking a bit more technically, in the MO framework
the probability that
the $i$-th insured does not survive one additional time unit is
$\mathbf{P}[X_i<1]=1-\exp\{-(\lambda_i+\lambda_{n+1})\}>1-\exp\{-\lambda_i\},\ i=1,\ldots,n$, where
the inequality holds intuitively because of the inclusion of the additional (common shock) r.f. that aggravates the marginal
mortality. The rationale for augmenting non-zero probabilities of simultaneous death
$\mathbf{P}[X_1=\cdots=X_n]$
comes naturally from, e.g.,  family
insurance coverages as a result of common exposure of the family members to aircraft and road accidents, among other
joint hazards. On a different note, the MO multivariate exponential distribution has been consistently
included in the professional actuarial exams and taught to students  (see, e.g., Bowers et al. (1997)).

Recently, the MO multivariate exponential distribution along with the underlying copula  (Nelsen (2006)) have been suggested as a reasonable
model for describing dependent defaults (Giesecke (2003)). In this respect, the r.v. $X_i$ is viewed as default time of the $i$-th
risk component (r.c.) $i=1,\ldots,n$ within the risk portfolio (r.p.) $\mathbf{X}$ exposed to $n$ idiosyncratic
and one systemic r.f.'s, having, correspondingly, exponentially distributed hitting times $E_{\lambda_i}$ and $E_{\lambda_{n+1}}$. In the context of default risk, the non-zero probabilities of simultaneous default are motivated empirically by, e.g.,  joint
insolvency of 24 railway firms that defaulted on June 21, 1970 (see, Azizpour \& Giesecke (2008)), whereas the exponentiality of the hitting times of distinct r.f.'s  follows naturally when risk occurrences are governed by $(n+1)$
stochastically independent Poisson processes.

In general, the exponentiality assumption on $E_{\lambda_j},\ j=1,\ldots,n+1$ mentioned above is deep-rooted in the modern
practical default modelling. Speaking briefly, the CreditRisk$^+$ approach to modelling dependent defaults
 (see, Frey \& McNeil (2003))
assumes that, for
$i=1,\ldots,n$, the stochastically independent hitting times ${}_iE_{\lambda_j}$ are  exponentially distributed
conditionally on the r.v.'s $\Lambda_1,\ldots,\Lambda_{n+1}$, which are mutually independent, gamma distributed and explain
the uncertainty
associated with the corresponding r.f.'s $j\in\{1,\ldots,n+1\}$. In such a case, the probability that the $i$-th r.c. defaults before the end of on arbitrary
time unit is given by
$\mathbf{P}[X_i<1|\ \boldsymbol{\Lambda}]=1-\exp\{-\sum_{j=1}^{n+1} w_j \Lambda_j \}$, where
$\boldsymbol{\Lambda}:=(\Lambda_1,\ldots,\Lambda_{n+1})'$ is a r.v. and $w_1.\ldots,w_{n+1}$ are deterministic positive weights.
The distributional properties of the CreditRisk$^+$ approach were explored in, e.g., Su \& Furman (2016).

In the present paper we introduce and study what can be viewed as an extension of the CreditRisk$^+$ approach.
Interestingly, on the one hand the structures proposed herein
can be motivated by the multi-factor ERM framework (Sweeting (2011)), and on the other hand they are of interest in
distribution theory, as we end up with a new multivariate probability distribution having Pareto of the 2nd kind univariate
margins, which unifies many existing stand alone results in, e.g., Arnold (1983, 2015), Chiragiev \& Landsman (2009), Asimit et al. (2010) and Su \& Furman (2016).
We note in passing in this respect that distributions with Paretian tails have been applied in a multitude of areas. We refer
to: Benson et al. (2007) - for applications in catastrophic events;  to Koedijk et al (1990), Longin (1996), Gabaix et al. (2003) - for applications in general financial phenomena; to Cebri\'{a}n et al. (2003) - for applications in insurance pricing; and to Soprano et al. (2010), Chavez-Demoulin et al. (2015) for applications in risk management.

In the sequel, we denote by $(X_1,\ldots,X_n)'$ the default times of generic financial units with labels in the set $\{1,\ldots,n\}$,
exposed to
$d(\in\mathbf{N})$ exogenous r.f.'s having stochastically independent hitting times $E_{\lambda_j}$ that are
exponentially distributed conditionally on $\Lambda_j,\ j=1,\ldots,d$. We assume that the r.v.'s $\Lambda_1,\ldots,\Lambda_{d}$
are gamma distributed and stochastically independent as in CreditRisk$^+$. The exposure matrix $c\in Mat_{n\times d}(\mathbf{1})$
having  entries $c_{i,j}\in \mathbf{1}:=\{0,\ 1\},\ i=1,\ldots,n,\ j=1,\ldots,d$  is deterministic and determined by the upper
management. Then we say that the $j$-th r.f. hits the set of r.c.'s $\mathcal{RC}_j:=\{i\in\{1,\ldots,n\}: c_{i,j}=1\}$ for any
$j=1,\ldots,d$. In a similar fashion, we say that the $i$-th r.c. is hit by the set of r.f.'s $\mathcal{RF}_i:=\{j\in\{1,\ldots,d\}: c_{i,j}=1
\}$ for any $i=1,\ldots,n$. The rest of the paper is organized as follows: We revisit and generalize when necessary some existing multivariate
dependence structures that are relevant to our present work in Section \ref{Sec-2}, and we then introduce and study the main object of interest herein -
the Multiple Risk
 Factor (MRF) dependence models in Section \ref{SecMRF}. In Section \ref{Sec-4} we specialize the discussion to the bivariate case, and
 in Section \ref{Sec-5} we exemplify the usefulness of the MRF dependencies within the contexts of financial risk measurement. In Section \ref{Sec-6} we elucidate the MRF dependencies with the help of a numerical example borrowed from the context of default risk. Section \ref{Sec-7} concludes the paper. Proofs are relegated to Appendix \ref{app-proof}.
 We note in passing that `$\overset{d}{=}$' stands for equality in distribution throughout the paper.

\end{section}

\begin{section}{Relevant existing dependencies revisited}
\label{Sec-2}
We start off with a number of simple but important notes. First, in what follows we
denote by $\Lambda\sim Ga(\alpha,\ \sigma)$ a r.v. that follows gamma distribution and has
the probability density function (p.d.f.)
\[
f_\Lambda(x)=\frac{1}{\Gamma(\alpha)}e^{-\sigma x}x^{\alpha-1}\sigma^{\alpha},\
x\in\mathbf{R}_+,
\]
where $\alpha\in\mathbf{R}_+$ and $\sigma\in\mathbf{R}_+$ are deterministic shape and rate
parameters, respectively. Second, we denote by $E_\lambda\sim Exp(\lambda)$ an
exponentially distributed r.v. with the p.d.f.
\[
f_{E_\lambda}(x)=\lambda e^{-\lambda x},\ x\in\mathbf{R}_+,
\]
where $\lambda\in\mathbf{R}_+$ is a deterministic parameter.
Third, we denote by `$*$' the mixture operator such that given two  appropriately jointly measurable r.v.'s $X_\beta\sim C(\beta)$ with $\beta\in \mathcal{B}\subseteq\mathbf{R}$ and
$B\sim H$, it holds that $X_\beta *B$ has the same distribution as $X_B$, or
succinctly $X_\beta *B\overset{d}{=}X_B$, where we of course assume that
the r.v. $B$ has its range in $\mathcal{B}$.
Then we readily have that
\[
E_{\Lambda}\overset{d}{=}E_\lambda * \Lambda\sim Pa(II)(\sigma,\ \alpha),
\]
that is $E_{\Lambda}$ follows the Lomax distribution. We note in this respect that a
Lomax distributed r.v.
$X\sim Pa(II)(\sigma,\ \alpha)$ has the p.d.f. and the decumulative distribution function (d.d.f.), respectively,
\[
f_X(x)=\frac{\alpha}{\sigma}\left(1+\frac{x}{\sigma}\right)^{-\alpha-1}
\textnormal{ amd }
\overline{F}_X(x):=\mathbf{P}[X>x]=\left(1+\frac{x}{\sigma}\right)^{-\alpha},\ x\in\mathbf{R}_+,
\]
where $\sigma\in\mathbf{R}_+$ is a scale parameter and $\alpha\in\mathbf{R}_+$ is a power
parameter. Last but not least, the classical extension of the univariate Lomax (and more generally
Pareto of the 2nd kind) distributions to the multivariate context is arguably the one of Arnold (1983, 2015)
with the d.d.f.
\begin{equation}
\label{ddfArnold}
\overline{F}(x_1,\ldots,x_n):=\mathbf{P}[X_1>x_1,\ldots,X_n>x_n]=\left(
1+\sum_{i=1}^n \frac{x_i}{\sigma_i}
\right)^{-\alpha},\ (x_1,\ldots,x_n)'\in\mathbf{R}_+^n,
\end{equation}
and the Pearson correlation for, $1\leq i\neq k\leq n$ and $\alpha>2$,
\begin{equation}
\label{rhoArnold}
\mathbf{Corr}[X_i,X_k]\equiv\frac{1}{\alpha}\in(0,\ 1/2).
\end{equation}

Two recently introduced multivariate probability distributions are of central interest for the forthcoming derivations. Namely, these are the dependence structures with univariate Pareto margins introduced in Asimit et al. (2010) and in Su \& Furman (2016).
Further we discuss  the just mentioned probability distributions and generalize them slightly where needed.

\begin{subsection}{The model of Asimit et al. (2010)}
\label{Subsec_Asimit}
Recall that $E_{\lambda_j},\ j=1,\ldots,n+1$ denote $(n+1)$ stochastically independent exponentially distributed r.v.'s
with parameters $\lambda_j\in\mathbf{R}_+$ and $n\in\mathbf{N}$. Further let $\Lambda_j\sim Ga(\alpha_j,\ 1)$
be gamma distributed r.v.'s all  independent
mutually and on $E_{\lambda_j},\ j=1,\ldots,n+1$. Then the r.v. $\mathbf{X}=(X_1,\ldots,X_n)'$ is said to follow
the multivariate probability distribution of Asimit et al. (2010) if its coordinates admit the following
stochastic representation
\begin{equation}\label{stochAsimit}
X_i\overset{d}{=} \sigma_i \left(
E_{\Lambda_i} \wedge E_{\Lambda_{n+1}}
\right)
\end{equation}
with the corresponding joint d.d.f. given by
\begin{equation}\label{ddfAsimit}
\overline{F}(x_1,\ldots,x_n)=\left(1+\bigvee_{i=1}^n \frac{x_i}{\sigma_i}
\right)^{-\alpha_{n+1}}
\prod_{i=1}^n
\left(
1+\frac{x_i}{\sigma_i}
\right)^{-\alpha_i},\ (x_1,\ldots,x_n)'\in\mathbf{R}^n_{+}.
\end{equation}

Speaking practically, if $E_{\lambda_j},\ j=1,\ldots,n+1$ are random and exponentially
distributed (default) times,
then (\ref{stochAsimit}) and (\ref{ddfAsimit}) are akin to the famous common shock set-up (Bowers et al. (1997))
with $(n+1)$ exogenous risk factors, of which $n$ are idiosyncratic and one is the common shock. By construction, the times of occurrence of the idiosyncratic risks are independent, whereas the times of occurrence of the systemic risks are fully-commonotonic. The presence of the common shock
risk factor results in a singularity, that is $\mathbf{P}[X_1=\cdots=X_n]=\alpha_{n+1}/(\alpha_1+\cdots+\alpha_n+\alpha_{n+1})$. Thereby, d.d.f.
(\ref{ddfAsimit}) is not absolutely continuous with respect to the Lebesgue measure on $\mathbf{R}_+^n$.
The presence of singularity is in general desirable as it addresses the  phenomenon of simultaneous death/default,
but at the same time it reduces the tractability of (\ref{ddfAsimit}) considerably (Asimit et al. (2010, 2016)).

A possibility of having numerous systemic r.f.'s that may hit subsets of the coordinates of the r.p. $\mathbf{X}$ as well as the
general lines of Sweeting (2011) motivate an extension of the set-up described above to allow for an arbitrary number, say $l(\in\mathbf{N})$ of r.f.'s. To this end, let
$c^l\in Mat_{n\times l}(\mathbf{1})$ be an $n\times l$ matrix with entries in $\{0,\ 1\}$. For the sake of the discussion in this section we assume that $n<l$. Also, let $\mathcal{RF}^l_i:=\{j\in\{1.\ldots,l\}:c^l_{i,j}=1\}$ denote the set of all risk factors
that `attack' risk component $i\in\{1,\ldots,n\}$.

\begin{definition}\label{AsimitGenDef}
Let $E_{\lambda_j},\ j=1,\ldots,l$ be stochastically independent exponentially distributed r.v.'s
with parameters $\lambda_j\in\mathbf{R}_+$. Also let $\Lambda_j\sim Ga(\alpha_j,\ 1)$
be gamma distributed r.v.'s all  independent
mutually and on $E_{\lambda_j},\ j=1,\ldots,l$. Set, for $\sigma_i\in\mathbf{R}_+$,
\[
X_i=\sigma_i \bigwedge_{j\in\mathcal{RF}^l_i} E_{\Lambda_j} ,
\]
then the joint distribution of $X_1,\ldots,X_n$ is a generalized variant of the one introduced in Asimit et al. (2010).
Succinctly, we write $\mathbf{X}\sim Pa_{1\ldots,n}^{c_l}(II)(\boldsymbol{\sigma},\ \boldsymbol{\alpha})$, where $\boldsymbol{\sigma}:=(\sigma_1,\ldots,\sigma_n)'$ is an $n$-dimensional vector of scale parameters and $\boldsymbol{\alpha}:=(\alpha_1,\ldots,\alpha_{l})'$ is an $l$-dimensional  vector of power parameters, all the parameters are positive
reals.
\end{definition}

Clearly the distribution of the $i$-th coordinate is Lomax, i.e., we have that $X_i\sim Pa(II)(\sigma_i,\ \alpha_{c,i})$ where
$\alpha_{c,i}:=\sum_{j\in\mathcal{RF}_i^l} \alpha_j,\ i=1,\ldots,n$. The joint distribution is given in the following theorem. Therein
we let $\alpha_{c,(i,k)}:=\sum_{j\in\mathcal{RF}_{(i,k)}^l} \alpha_j$, where
$\mathcal{RF}_{(i,k)}^l:=\{j\in\{1,\ldots,l\}:\ c_{k,j}^lc_{i,j}^l=1\}$, that is the set
$\mathcal{RF}_{(i,k)}^l$ contains the risk factors that attack both the $i$-th and
the $k$-th risk components, $1\leq i\neq k\leq n$. The proof of the theorem is omitted, as it is readily obtained along the lines of  the proofs of Propositions
2.3 and 3.2 in Asimit et al. (2010).

\begin{theorem}\label{AsimitGenProp}
Let $\mathbf{X}\sim Pa_{1\ldots,n}^{c^l}(II)(\boldsymbol{\sigma},\ \boldsymbol{\alpha})$,
then its d.d.f. is given by
\begin{equation}
\label{ddf_c}
\overline{F}(x_1,\ldots,x_n)=\prod_{j=1}^{l}\left(1+\bigvee_{i\in\mathcal{RC}_j}
\frac{x_i}{\sigma_i}\right)^{-\alpha_j},\ (x_1,\ldots,x_n)'\in\mathbf{R}_+^n.
\end{equation}
Moreover, the index of the Pearson correlation is given by
\begin{equation}
\label{rhoAsimitformula}
\mathbf{Corr}[X_i,X_k]=\frac{\alpha_{c,(i,k)}}{\alpha_{c,i}+\alpha_{c,k}-\alpha_{c,(i,k)}-2}
\sqrt{\frac{(\alpha_{c,i}-2)(\alpha_{c,k}-2)}{\alpha_{c,i} \alpha_{c,k}}},
\end{equation}
subject to $\alpha_{c,i}>2$ and $\alpha_{c,k}>2$ with $1\leq i\neq k\leq n$.
\end{theorem}

A remark about the $c^l$-matrix is perhaps well-timed now. It is useful to realize that the importance of the matrix  is in that
it shapes the set of risk factors. More specifically, a risk factor $j\in\{1,\ldots,l\}$ is idiosyncratic if and only if
$\sum_{i=1}^n c_{i,j}^l=1$, and it is common shock otherwise. Moreover, in the latter case the risk factor
($j\in\{1,\ldots,l\}$ again) affects the subset  of  risk components
$\mathcal{RC}_j=\{i\in\{1,\ldots,n\}:\ c_{i,j}^l= 1\}$.

We conclude this subsection by noticing that the idea of randomization, and Definition
\ref{AsimitGenDef} is nothing but that, has been well-studied and can be traced back to Feller (1966) -
probability theory and Bunke (1969) - economics. Important recent references are Geweke \& Amisano (2011)  -  economics, McNeil et al. (2005) and references therein - quantitative risk, generally and Gordy (2000) - credit risk, particularly. Also, as  (\ref{rhoAsimitformula}) can take on any
value in $[0,\ 1]$,  d.d.f. (\ref{ddf_c}) covers the entire range of positive dependence when measured by the Pearson index of
correlation. We show in the next subsection that neither one of the dependence structures
investigated in Arnold (1983, 2015) and Chiragiev \& Landsman (2009) enjoys the just mentioned important property.
\end{subsection}

\begin{subsection}{The model of Su \& Furman (2016)}
\label{Subsec_Su}

Let $E_{\lambda_j},\ j=1,\ldots,n+1$ denote as hitherto exponentially distributed and stochastically independent r.v.'s
with parameters $\lambda_j\in\mathbf{R}_+$ and $n\in\mathbf{N}$. Further let ${}_iE_{\lambda_j},\ i=1,\ldots,n$
be independent copies of the r.v. $E_{\lambda_j},\ j=1,\ldots,n+1$. In addition,
let $\Lambda_{j}\sim Ga(\gamma_j,1)$ be gamma distributed r.v.'s independent mutually and on
${}_iE_{\lambda_{j}},\ i=1,\ldots,n,\ j=1,\ldots,n+1$. The r.v. $\mathbf{X}=(X_1,\ldots,X_n)'$ follows the multivariate probability distribution of Su \& Furman (2016) if its coordinates admit the following stochastic representation
\begin{equation}\label{StocSu}
X_i\overset{d}{=}\sigma_i\left(
{}_iE_{\Lambda_{i}} \wedge {}_iE_{\Lambda_{n+1}}
\right)
\end{equation}
where $\sigma_i\in\mathbf{R}_+$ for $i=1,\ldots,n$, and the joint d.d.f. is given by
\begin{equation}\label{DdfSu}
\overline{F}(x_1,\ldots,x_n)=\prod_{i=1}^{n}\left(
1+\frac{x_i}{\sigma_i}
\right)^{-\gamma_i}
\left(
1+\sum_{i=1}^{n} \frac{x_i}{\sigma_i}
\right)^{-\gamma_{n+1}}, \ (x_1,\ldots,x_n)'\in\mathbf{R}_+^n.
\end{equation}

Speaking plainly, (\ref{StocSu}) and (\ref{DdfSu}) mimic the common shock framework described in Subsection
\ref{Subsec_Asimit}, but circumvent the emergence of singularity  by assuming that the  times of occurrence of the
common shock risk factors are not fully-comonotonic but rather conditionally independent. As in Subsection \ref{Subsec_Asimit},
we extend (\ref{StocSu}) and (\ref{DdfSu}) to augment $m(\in\mathbf{N})$
risk factors. To this end, let $c^m$ be an $n\times m$ deterministic matrix with the entries equal to either
zero or one, and let $\mathcal{RF}^m_i:=\{ j\in\{1,\ldots,m\}:\ c^m_{{i,j}}=1\}$ for $i=1,\ldots,n$.

\begin{definition}\label{SuGenDef}
Let ${}_iE_{\lambda_{j}},\ i=1,\ldots,n$ be independent copies of the stochastically independent
exponentially distributed r.v.'s $E_{\lambda_j},\ \lambda_{j}\in\mathbf{R}_+,\ j=1,\ldots,m$.
Also let $\Lambda_j\sim Ga(\gamma_j,\ 1)$
be gamma distributed r.v.'s all  independent
mutually and on ${}_iE_{\lambda_{j}},\ i=1,\ldots,n,\ j=1,\ldots,m$. Set, for $\sigma_i\in\mathbf{R}_+$,
\[
X_i=\sigma_i \bigwedge_{j\in\mathcal{RF}_i^m} {}_iE_{\Lambda_{j}},
\]
then the joint distribution of $X_1,\ldots,X_n$ is a generalized variant of the one introduced in Su and Furman (2016). Succinctly, we write $\mathbf{X}\sim Pa_{1\ldots,n}^{c^m}(II)(\boldsymbol{\sigma},\ \boldsymbol{\gamma})$, where $\boldsymbol{\sigma}:=(\sigma_1,\ldots,\sigma_n)'$ is an $n$-dimensional vector of scale parameters and $\boldsymbol{\gamma}:=(\gamma_{1},\ldots,\gamma_m)'$ is an $m$-dimensional  vector of power parameters, all parameters are positive reals.
\end{definition}

The distribution of $X_i,\ i=1,\ldots,n$ is clearly $Pa(II)(\sigma_i,\ \gamma_{c,i})$, where $\gamma_{c,i}=
\sum_{j\in\mathcal{RF}_{i}^m}  \gamma_j$, that is the coordinates of the random vectors in Definition
\ref{SuGenDef} are Lomax distributed, similarly to the situation in Definition
\ref{AsimitGenDef}. The
joint distribution of $X_1,\ldots,X_n$ is formulated in the following theorem.  We note in passing that by analogy with the discussion in the previous subsection, we let
$\gamma_{c,(i,k)}:=\sum_{j\in\mathcal{RF}_{(i,k)}^m} \gamma_j$ where
$\mathcal{RF}_{(i,k)}^m:=\{j\in\{1,\ldots,m\}:\  c^m_{i,j}c^m_{k,j}=1\}$.
In addition the $(q+1)\times q$ hypergeometric function is given by (Gradshteyn \& Ryzhik (2014))
\begin{eqnarray}
\label{hyperpq}
_{q+1}F_q(a_1,\ldots,a_{q+1};b_1,\ldots,b_q;z):=\sum_{k=0}^{\infty}\frac{(a_1)_k,\ldots,(a_{q+1})_k }{(b_1)_k,\ldots,(b_q)_k}\frac{z^k}{k!},\
\end{eqnarray}
where $(p)_n:=p(p+1)\ldots(p+n-1)$ for $n\in \mathbf{Z}_+$, $(p)_0:=1$ and $q\in\mathbf{Z}_+$.
For $a_1,\ldots,a_{q+1}$ all positive, and these are the cases of interest in the present paper,
the radius of convergence of the series is the open disk $|z|<1$. On the boundary $|z|=1$, the series
converges absolutely if $d=b_1+\cdots + b_q-a_1-\cdots -a_{q+1}>0$, and it
converges except at $z=1$ if $0\geq d>-1$. We omit the proof of the theorem as it is very similar to the proof of
Theorem 2.1 in Su \& Furman (2016).

\begin{theorem}\label{AsimitGenProp}
Let $\mathbf{X}\sim Pa_{1\ldots,n}^{c^m}(II)(\boldsymbol{\sigma},\ \boldsymbol{\gamma})$,
then its d.d.f. is given by
\begin{equation}
\label{ddf-ac}
\overline{F}(x_1,\ldots,x_n)=\prod_{j=1}^{m}\left(1+\sum_{i\in\mathcal{RC}_j}\frac{x_i}{\sigma_i}\right)^{-\gamma_j},\
(x_1,\ldots,x_n)'\in\mathbf{R}_+^n.
\end{equation}
Furthermore, the index of the Pearson correlation is given, for $\gamma_{c,i}>2$ and $\gamma_{c,k}>2$, by
\begin{equation}
\label{rhoSuformula}
\mathbf{Corr}[X_i,X_k]=\left(_3F_2\left(\gamma_{c,(i,k)},1,1;\gamma_{c,i},\gamma_{c,k};1 \right)-1\right)\sqrt{\frac{(\gamma_{c,i}-2)(\gamma_{c,k}-2)}{\gamma_{c,i} \gamma_{c,k}}}
\in[0,\ 1/2),
\end{equation}
where $1\leq i\neq k\leq n$. The hypergeomteric function in (\ref{rhoSuformula}) converges absolutely.
\end{theorem}

D.d.f. (\ref{ddf-ac}) is clearly absolutely continuous with respect to the corresponding Lebesgue measure, and
it unifies the classical multivariate Pareto distributions of Arnold (1983, 2015) as well as the recently introduced ones of
Chiragiev \& Landsman (2009). Also, the fact that (\ref{rhoSuformula}) is in $[0,1/2)$ is rather
unfortunate, as it makes d.d.f. (\ref{ddf-ac}) inappropriate for applications when
the r.v.'s $X_1,\ldots,X_n$ have strong positive correlations. Since d.d.f. (\ref{ddf-ac})
serves as the basic building block of the popular in the modern practical default risk measurement CreditRisk$^+$ approach, the aforementioned limited range of attainable
Pearson correlations conforms well to the empirical evidence showing that CreditRisk$^+$
may underestimate defaults' clustering (Das et al. (2007)).
\end{subsection}

To summarize the developments hitherto, we note that
\begin{itemize}
\item the restatements of the dependency models of Asmitit et al. (2011) and Su \& Furman (2016) in terms of times of occurrence with random risk factors  align well with the general set-up of default risk
and as such are well suited for the corresponding practical applications and can be easily conveyed to upper management;
\item simple stochastic representations
in Definitions \ref{AsimitGenDef} and \ref{SuGenDef} are essential and allow for convenient simulation analysis and
stresstesting;
\item d.d.f.'s
(\ref{ddf_c}) and (\ref{ddf-ac})
are analytically tractable and quite general unifying a variety of existing multivariate structures.
\end{itemize}

The advantages described above lay the groundwork for gluing the objects in Definitions \ref{AsimitGenDef} and \ref{SuGenDef} into one
general multivariate probability structure that would augment random times of occurrence that are a.)
 stochastically
independent, b.) conditionally independent, and c.) fully-comonotonic as well as
inherit the full range of attainable non-negative Pearson correlations.
This new general structure is referred to as the Multiple Risk Factor
dependence structure, and it is the main object of study
in the present paper.
\end{section}

\begin{section}{Multiple risk factor dependence structure}
\label{SecMRF}

Naturally, the gap between the stochastically independent and fully-comonotonic times of occurrence of r.f.'s in Definition \ref{AsimitGenDef} is somewhat
too harsh. Next we fill the gap by unifying
 Definitions \ref{AsimitGenDef} and \ref{SuGenDef}.

 Let $\mathbf{Y}:=(Y_1,\ldots,Y_n)'$ and $\mathbf{Z}:=(Z_1,\ldots,Z_n)'$ be two r.v.'s of dimension $n(\in\mathbf{N})$.
Next we formulate the MRF models of interest in this paper.

\begin{definition}
\label{MRF-stoch-Def}
Assume that $\mathbf{Y}\sim Pa_{1,\ldots,n}^{c^l}(II)(\boldsymbol{\sigma},\ \boldsymbol{\alpha})$ as in Definition \ref{AsimitGenDef}
and  $\mathbf{Z}\sim Pa_{1,\ldots,n}^{c^m}(II)(\boldsymbol{\sigma},\ \boldsymbol{\gamma})$ as in Definition \ref{SuGenDef}
are stochastically independent, then the r.v.
$\mathbf{X}=(X_1,\ldots,X_n)'$ with each coordinate $X_i=Y_i\wedge Z_i,\ i=1,\ldots,n$
is said to follow a Multiple Risk Factor  dependence with Pareto of the 2nd kind univariate marginal distributions; notationally
$\mathbf{X}\sim Pa_{1,\ldots,n}^c(II)(\boldsymbol{\sigma},\ \boldsymbol{\xi})$. Here
 $c=(c^l,\ c^m)$ is an $n\times (l+m)$-dimensional block matrix,
$\boldsymbol{\sigma}$ is an $n$-dimensional vector of scale parameters and
$\boldsymbol{\xi}$ is an $(l+m)$- dimensional vector of power parameters, such that
$\xi_j=\alpha_j,\ j=1,\ldots,l(\in\mathbf{N})$ and $\xi_j=\gamma_j,\ j=l+1.\ldots,l+m(\in\mathbf{N})$.
\end{definition}
\noindent In what follows, we sometimes write $\mathbf{X}\sim Pa_{1,\ldots,n}^{c}(II)(\boldsymbol{\sigma},\ \boldsymbol{\alpha},\
\boldsymbol{\gamma})$, if we wish to emphasize the presence of $\boldsymbol{\gamma}$
and $\boldsymbol{\alpha}$ vectors of parameters.

It is clear that Definition \ref{MRF-stoch-Def} unifies multiple risk factor dependence
 structures with and without singularities into
one encompassing structure. One practical motivation for the MRF structures stems from Sweeting (2011)
(also, Duffie \& Singleton (1999)). More specifically, consider
a portfolio of default times of $n$  business units. Further assume that each component of the portfolio is exposed to a set of possibly $(l+m)$ fatal
risk factors, of which some are idiosyncratic - $\left\{j\in\{1,\ldots, l+m\}:\sum_{i=1}^n c_{i,j}=1\right\}$,
others are systemic with fully-commonotonic occurrence times - $\left\{j\in\{1,\ldots, l\}:\sum_{i=1}^n c_{i,j}>1\right\}$, and there exists yet another group of risk factors
$\left\{j\in\{l+1,\ldots, l+m\}:\right.$ $\left.\sum_{i=1}^n c_{i,j}>1\right\}$ that attack the risk components at positively dependent
but not perfectly dependent times. Then, assuming that the hitting times of the risk factors are stochastically independent group-wise and
also that the risk factors are gamma distributed as in CreditRisk$^+$, the MRF structure \`a la  Definition \ref{MRF-stoch-Def} is obtained. Another motivation for the MRF
dependencies stems from the recent trends of prudence that have taken the modern financial
risk measurement by storm. Indeed, multivariate probability distributions having heavy tailed
and positively dependent  univariate margins are in a significant
practical demand
nowadays.

The joint d.d.f. of
the MRF r.v.'s is given in the next theorem. Recall that $\mathcal{RC}_j$ has been defined
as the set that contains all the risk components that are exposed to the $j$-th risk
factor, $j\in\{1,\ldots,l+m\}$. The proof of the theorem is simply by construction of the MRF r.v.'s, and it is thus
omitted.

\begin{theorem}
\label{def-gp}			
Let $\mathbf{X}\sim Pa_{1,\ldots,n}^c(II)(\boldsymbol{\sigma},\ \boldsymbol{\xi})$, then the joint d.d.f. of
its coordinates is given by
\begin{equation}
\label{ddf-gp}
\overline{F}(x_1,\ldots,x_n)=\prod_{j=1}^{l}\left(1+\bigvee_{i\in\mathcal{RC}_j}
\frac{x_i}{\sigma_i}\right)^{-\xi_j}
\prod_{j=l+1}^{l+m}\left(1+\sum_{i\in\mathcal{RC}_j}\frac{x_i}{\sigma_i}\right)^{-\xi_j},\
(x_1,\ldots,x_n)'\in\mathbf{R}_+^n,
\end{equation}
where $\boldsymbol{\sigma}$ is an $n$-dimensional vector of scale parameters, and $\boldsymbol{\xi}$ is an $(l+m)$- dimensional vector of power parameters.
\end{theorem}

On the one hand, d.d.f. (\ref{ddf-gp}) can be viewed as a new multivariate probability
 distribution
with Lomax distributed univariate margins, which, for appropriately chosen $c$
matrices, reduces to, e.g., the classical model of Arnold (1983, 2015), as well as to these of Chiragiev \& Landsman (2009),
Asimit et al. (2010) and Su \& Furman (2016). As such it can be used to describe dependent
(actuarial) risks with heavy tailed marginal behaviour. On the other hand,
d.d.f. (\ref{ddf-gp}) enjoys the interpretation of a dependent default times
model that encompasses an arbitrary number of exogenous risk factors having
stochastically independent, positively orthant dependent (see, Joe (1997) and Denuit et al. (2005)) or fully-comonotonic
occurrence times. In either case, the following theorem is of basic importance.
The proof is simple and thus omitted.

\begin{theorem}
\label{marginal-p}
Let $\mathbf{X}\sim Pa_{1,\ldots,n}^c(II)(\boldsymbol{\sigma},\ \boldsymbol{\xi})$, then, for $i=1,\dots,n$, we have that
each coordinate is Lomax distributed, i.e., $X_i\sim Pa(II)(\sigma_i,\ \xi_{c,i}),$ where $\xi_{c,i}=\sum_{j\in\mathcal{RF}_i}\xi_j $
and therefore
\begin{itemize}
\item the d.d.f. of $X_i$ is
\[
\overline{F}(x)=\left(1+\frac{x}{\sigma_i}\right)^{-\xi_{c,i}},\ x\in\mathbf{R}_+;
\]
\item the mathematical expectation of $X_i$ is, for $\xi_{c,i}>1$,
\[
\mathbf{E}[X_i]=\frac{\sigma_i}{\xi_{c,i}-1};
\]
\item the variance of $X_i$ is, for $\xi_{c,i} >2$,
\[
\mathbf{Var}[X_i]=\frac{\sigma_i^2\xi_{c,i}}{(\xi_{c,i} -1)^2(\xi_{c,i}-2)}.
\]
\end{itemize}
\end{theorem}
\noindent  We note in passing that the expectation/variance of the r.v. $X_i,\ i=1,\ldots,n$
 can be finite even if the expectations/variances of occurrence
times of some risk factors $j\in\{1,\ldots,l+m\}$ are infinite.

When the r.v.'s $X_1,\ldots,X_n$ denote times of default of various financial units $i\in\{1,\ldots,n\}$, the r.v.'s
$X_-:=\wedge_{i=1}^n X_i$ and $X_+:=\vee_{i=1}^n X_i$ denote, respectively,  first and  last default times.
We say that $X_-$ and $X_+$ have distributions
$F_-$ and $F_+$, respectively, and they play a pivotal role in the general
theory of credit
risk (Adalsteinsson (2014)) and in the mathematics of life contingencies (Bowers et al. (1997)).
We further investigate the d.d.f.'s of the two. In this respect, the following lemma is of central
importance.

\begin{lemma}
\label{FL2005} For $i=1,\ldots,n$, let $\Lambda_i\sim Ga(\gamma_i(\in\mathbf{R}_+),\ \sigma_i(\in\mathbf{R}_+)),$ denote independent
gamma distributed r.v.'s, and let $\Lambda^\ast :=\Lambda_1+\cdots +\Lambda_n$ be their sum.
Then  $\Lambda^\ast\sim Ga(\gamma^\ast+K,\ \sigma_{+}),$ where
$\gamma^\ast=\gamma_1+\cdots+\gamma_n$, $\sigma_{+}=\vee_{i=1}^n \sigma_i$ and $K$ is an
integer-valued non-negative r.v. with the probability mass function given, for $k\in\mathbf{N}_0$,  by
\begin{equation}
\label{pk}
p_k=\mathbf{P}[K=k]=c_{+} \delta_k
\end{equation}
where
\[
c_{+}=\prod_{i=1}^n \left(\frac{\sigma_i}{\sigma_{+}}\right)^{\gamma_i}
\]
and
\[
\delta_k=k^{-1}\sum_{l=1}^k\sum_{i=1}^n \gamma_i \left(
1-\frac{\sigma_i}{\sigma_{+}}
\right)^l\delta_{k-l}, \textnormal{ for } k>0 \textnormal{ and } \delta_0=1.
\]
Moreover, the distribution of $E_{\Lambda^\ast}$ is $Pa(II)(\sigma_+,\ \gamma^\ast+K)$.
\end{lemma}

We are now ready to prove that in the context of the MRF dependencies, the distribution of the first default time is Lomax with a random power parameter.
\begin{theorem}
\label{pro-min}
Let
$\mathbf{X}\sim Pa_{1,\ldots,n}^c(II)(\boldsymbol{\sigma},\ \boldsymbol{\xi})$, then $X_-\sim Pa(II)(\sigma_+,\xi^\ast+K)$, where
\[
\sigma_+=\left(
\bigvee_{j=1}^{l} \bigwedge_{i\in\mathcal{RC}_j}\sigma_i
\right)
\vee
\left(
\bigvee_{j=l+1}^{l+m}\left(\sum_{i\in\mathcal{RC}_j}\frac{1}{\sigma_i}\right)^{-1}
\right),
\]
and $\xi^\ast+K=\xi_1+\cdots+\xi_{l+m}+K$ is a random power parameter with the integer valued r.v. $K$ having probability mass function
(\ref{pk}).
\end{theorem}

The next assertion is similar to Proposition 2 of Vernic (2011) as well as to
Proposition 2.4 of Su \& Furman (2016).

\begin{corollary}
\label{max-th}
Let $\mathbf{X}\sim Pa_{1,\ldots,n}^c(II)(\boldsymbol{\sigma},\ \boldsymbol{\xi})$, then the d.d.f. of the last default r.v. is a linear
combination of the d.d.f.'s of univariate Pareto of the 2nd kind r.v.'s with random power parameters, namely
\begin{equation}
\label{maxddf}
\overline{F}_+(x)=\sum_{\mathcal{S}\subseteq\{1,\ldots,n\}} (-1)^{|\mathcal{S}|-1}
\overline{F}_{\mathcal{S}-}(x),\ x\in\mathbf{R}_+,
\end{equation}
where $X_{\mathcal{S}-}=\wedge_{s\in\mathcal{S}\subseteq\{1,\ldots,n\}}X_s$ and
$X_{\mathcal{S}-}\sim F_{\mathcal{S}-}$.
\end{corollary}

We conclude this section by stating yet another important property of the MRF dependence structures, which arguably distinguishes the construction
from the majority of the multivariate probability distributions existing nowadays. In this respect, let $k\leq n$ and let $i_1,\ldots,i_k\in\{1.\ldots,n\}$
establish an index set, then the MRF dependence structures allow the probability $\mathbf{P}[
X_{i_1}=\cdots=X_{i_k}
]$ to be non-zero. This phenomenon, which can be easily motivated in practice by, e.g., the presence of a parent subsidiary or similar contractual
relationships - in the context of default risk, and by occurrence of a catastrophe that affects a number of lives simultaneously - in the context of life contingencies,
is possible because the MRF distributions may have singularities and thus are in general not absolutely continuous with respect to the Lebesgue
measure.
In the next theorem we study the probability mass that is assigned to the singularity.

Before stating and proving the next result, we extend some of the notations already used. First, the set of all risk factors that attack sub-portfolio
$(X_{i_1},\ldots,X_{i_k})'$ is, for $k\geq 2$, in the sequel  denoted by
\begin{equation}
\label{RFall}
\mathcal{RF}_{i_1,\ldots,i_k}:=\left\{j\in\{1,\ldots,l+m\}:\ c_{i_h,j}=1 \textnormal{ for at least one } i_h\in\{i_1,\ldots,i_k\}\right\},
\end{equation}
which is the union of two disjoint sets, that is of
\begin{equation}
\label{RFsing}
\mathcal{RF}_{(i_1,\ldots,i_k)}:=\{j\in\{1,\ldots,l+m\}:  c_{i_h,j}=1
\textnormal{ for all } i_h\in\{i_1,\ldots,i_k\}\}
\end{equation}
and
\begin{equation}
\label{RFnotsing}
\mathcal{RF}_{\overline{(i_1,\ldots,i_k)}}:=\mathcal{RF}_{i_1,\ldots,i_k}\setminus
\mathcal{RF}_{(i_1,\ldots,i_k)}.
\end{equation}
Also, let
\begin{equation}
\label{RFihnotsing}
\mathcal{RF}_{i_h,\overline{(i_1,\ldots,i_k)}}:=\mathcal{RF}_{i_h}\setminus \mathcal{RF}_{(i_1,\ldots,i_k)}.
\end{equation}
Based on the notation above, we can have, e.g.,
$\xi_{c,i_1,\ldots,i_k}=\sum_{j\in\mathcal{RF}_{i_1,\ldots,i_k}} \xi_j$ and
in a similar fashion
$\xi_{c,(i_1,\ldots,i_k)}=\sum_{j\in\mathcal{RF}_{(i_1,\ldots,i_k)}} \xi_j$,
$\xi_{c,\overline{(i_1,\ldots,i_k)}}=\sum_{j\in\mathcal{RF}_{\overline{(i_1,\ldots,i_k)}}} \xi_j$ as
well as
$\xi_{c,{i_h,\overline{(i_1,\ldots,i_k)}}}=\sum_{j\in\mathcal{RF}_{i_h,\overline{(i_1,\ldots,i_k)}}} \xi_j$. Of course,
by analogy, we may need to sum over the coordinates of $\boldsymbol{\gamma}$ or
 $\boldsymbol{\alpha}$, and in such cases we add superscripts `l' and `m', respectively,
 to sets
(\ref{RFall}) to (\ref{RFihnotsing}).
\begin{theorem}
\label{sim-default}
Let $\mathbf{X}\sim Pa_{1,\ldots,n}^c(II)(\boldsymbol{\sigma},\ \boldsymbol{\xi})$,
and let as before $i_1,\ldots,i_k\in\{1,\ldots,n\}$ for
$2\leq k\leq n$. Then we have that
\begin{equation}
\label{probsin}
\mathbf{P}[X_{i_1}/\sigma_{i_1}=\cdots= X_{i_k}/\sigma_{i_k}]=\alpha_{c,(i_1.\ldots,i_k)}\mathbf{E}\left[\frac{1}{\xi_{c,i_1,\ldots,i_k}+K}\right],
\end{equation}
where $K$ is an integer-valued r.v. with the p.m.f. \`a la (\ref{pk}).
\end{theorem}

Interestingly, if the r.v. $K$ is zero almost surely, then  (\ref{probsin}) reduces to
\[
\mathbf{P}[X_{i_1}/\sigma_{i_1}=\cdots =X_{i_k}/\sigma_{i_k}]=
\frac{\alpha_{c,(i_1,\ldots,i_k)}}{\xi_{c,i_1,\ldots,i_k}}.
\]
This happens when there is no risk factors with positively dependent (but not fully
commonotonic) default times present in the model. If in addition, $\mathcal{RC}_j$ is
an empty set for all $j=l+1,\ldots,l+m$,
then (\ref{probsin}) simplifies to
\[
\mathbf{P}[Z_{i_1}/\sigma_{i_1}=\cdots =Z_{i_k}/\sigma_{i_k}]=
\frac{\alpha_{c,(i_1,\ldots,i_k)}}{\alpha_{c,i_1,\ldots,i_k}},
\]
which is a slightly more general expression than the one obtained in Asimit et al. (2010).

In summary, it is instructive to note that the probability mass that is assigned to the
singular part of d.d.f. (\ref{ddf-gp}) is proportional to $\alpha_{c,(i_1,\ldots,i_k)}$.
In the language of default risk, this means that the stronger the contribution of the
risk factors having fully-comonotonic hitting times is, the higher the probability of
simultaneous default becomes.
\end{section}

\begin{section}{Multiple Risk Factor dependencies: bivariate case}
\label{Sec-4}
In this section we study some properties of the bivariate variant of the MRF structure.
For instance, we
derive joint product moments of $(X_i,X_k)'\sim Pa_{i,k}^c(II)(\boldsymbol{\sigma},\ \boldsymbol{\xi})$, where
$\boldsymbol{\sigma}=(\sigma_i,\ \sigma_k)'$ and $\boldsymbol{\xi}=(\xi_1,\ldots,\xi_{l+m})'=(\alpha_1,\ldots,\alpha_l,\gamma_{l+1},\ldots,\gamma_{l+m})'$ are two deterministic vectors
of scale and power parameters, respectively.

It is a simple matter to see that d.d.f. (\ref{ddf-gp}) reduces to
\begin{eqnarray}
\label{ddf-2}
&&\overline{F}(x,y) \label{ddfbiv} \\
&=&\left(
1+\frac{x}{\sigma_i}\bigvee \frac{y}{\sigma_k}
\right)^{-\alpha_{c,(i,k)}}
\left(
1+\frac{x}{\sigma_i}+\frac{y}{\sigma_k}
\right)^{-\gamma_{c,(i,k)}}
\left(
1+\frac{x}{\sigma_i}
\right)^{-\xi_{c,i,\overline{(i,k)}}}
\left(
1+\frac{y}{\sigma_k}
\right)^{-\xi_{c,k,\overline{(i,k)}}} \notag ,
\end{eqnarray}
for $1\leq i\neq k\leq n$ and $(x,y)'\in\mathbf{R}_+^2$.
The next theorem states the Lebesgue decomposition of the d.d.f. above.
\begin{theorem}
\label{decomp}
For $1\leq i\neq k\leq n$,
let $(X_i,X_k)' \sim Pa(II)^{c}_{i,k}(\boldsymbol{\sigma}, \boldsymbol{\alpha},\boldsymbol{\gamma})$,
 then its d.d.f. can be decomposed as
\begin{equation}
\label{decomp-formula}
\overline{F}(x,y)=a\overline{F}_s(x,y)+(1-a)\overline{F}_{ac}(x,y), \textnormal{ where } a\in[0,\ 1] \textnormal{ and }
(x,y)'\in\mathbf{R}_+^2.
\end{equation}
In (\ref{decomp-formula}), the singular component concentrates its mass on the line
$\left\{(x,y)'\in\mathbf{R}_+^2:\ \frac{x}{\sigma_i}=\frac{y}{\sigma_k}\right\}$ and is given by
\begin{equation}
\label{scomp}
\overline{F}_s(x,y)=\frac{\alpha_{c,(i,k)}}{a}\mathbf{E}\left[
\frac{1}{\xi_{c,i,k}+K}\left(1+\frac{x}{\sigma_i}\bigvee\frac{y}{\sigma_k}\right)^{-(\xi_{c,i,k}+K)}
\right],
\end{equation}
whereas the absolutely continuous component is given, for $(x,y)'\in\mathbf{R}_+^2$, by
\begin{equation}
\label{acomp}
\overline{F}_{ac}(x,y)=\frac{1}{1-a}\left(
\overline{F}(x,y)-a\overline{F}_s(x,y)
\right),
\end{equation}
where $a=\mathbf{P}[X_i/\sigma_i=X_k/\sigma_k]$ is given in Theorem \ref{sim-default}, $\xi_{c,i,k}=\sum_{j\in\mathcal{RF}_{i,k}} \xi_j$
and $K$ is an
integer-valued non-negative r.v. with a p.m.f. \`a la (\ref{pk}).
\end{theorem}

We further derive the Pearson index of correlation for a random pair with the MRF dependence. Speaking strictly,
the Pearson $\rho$ has been criticized by many authors, yet it remains a ubiquitous measure of correlation when it comes
to financial risk management and/or actuarial science. Let
\begin{equation}
\label{hfun}
h(x;a,b):={}_3F_2(x-1,1,a;x,b-1;-1).
\end{equation}
where ${}_3F_2$ is a special form of the hypergeometric function in (\ref{hyperpq}) and all of $x,\ a,\ b$ are positive reals. Then we have the following result that
formulates the joint product moment in the context of the MRF structures.
\begin{theorem}
\label{joint_m}Let $(X_i,X_k)' \sim Pa(II)^{c}_{i,k}(\boldsymbol{\sigma},\ \boldsymbol{\alpha},\ \boldsymbol{\gamma})$, then the product moment is, if finite, formulated as follows
\begin{eqnarray*}
\mathbf{E}[X_iX_k]=\sigma_1\sigma_2
\frac{(\xi_{c,k}-1)h(\xi_{c,i};\gamma_{c,(i,k)},\xi_{c,i,k})+(\xi_{c,i}-1)h(\xi_{c,k};\gamma_{c,(i,k)},
\xi_{c,i,k})}{(\xi_{c,i,k}-2)(\xi_{c,i}-1)(\xi_{c,k}-1)},
\end{eqnarray*}
where $\xi_{c,i}=\alpha_{c,i}+\gamma_{c,i},\ \xi_{c,k}=\alpha_{c,k}+\gamma_{c,k}$ and $\xi_{c,i,k}=\alpha_{c,i,k}+\gamma_{c,i,k}$.
If in addition, we have that $\xi_{c,i}>2,\ \xi_{c,k}>2$, then the Pearson correlation is
\begin{eqnarray*}
&&\mathbf{Corr}[X_i,X_k]
=\sqrt{\frac{(\xi_{c,i}-2)(\xi_{c,k}-2)}{\xi_{c,i}\xi_{c,k}}}\\
&\times&\frac{(\xi_{c,k}-1)h(\xi_{c,i};\gamma_{c,(i,k)},\xi_{c,i,k})+(\xi_{c,i}-1)h(\xi_{c,k};\gamma_{c,(i,k)},\xi_{c,i,k})-\xi_{c,i,k}+2}{(\xi_{c,i,k}-2)},
\end{eqnarray*}
where $1\leq i\neq k\leq n$.
\end{theorem}

Clearly, the Pearson correlation coefficients for the special cases of the MRF structures
having no risk factors with a.) fully comonotonic and
b.) positively dependent but not fully comonotonic hitting times are
recoverable from the general expression in Theorem \ref{joint_m}.
\begin{corollary}
\label{cov-su}
Let $\alpha_{c,(i,k)}\equiv 0$ and assume that the rest of the conditions in Theorem \ref{joint_m} hold, then
\begin{equation}
\label{cov-pos-gen}
\mathbf{Corr}[X_i,X_k]=\sqrt{\frac{(\xi_{c,i}-2)(\xi_{c,k}-2)}{\xi_{c,i}\xi_{c,k}}}
\left(_3F_2\left(\gamma_{c,(i,k)},1,1;\xi_{c,i},\xi_{c,k};1 \right)-1\right),
\end{equation}
where $1\leq i\neq k\leq n$.
\end{corollary}

\noindent Formula (\ref{cov-pos-gen}) recovers the one derived in Chiragiev \& Landsman (2009) as well as the recent one in
Su \& Furman (2016). In the simplest case, i.e., when $\alpha_{c,(i,k)},\ \alpha_{c,i,\overline{(i,k)}},\ \alpha_{c,k,\overline{(i,k)}}$
and also $\gamma_{c,i,\overline{(i,k)}},\ \gamma_{c,k,\overline{(i,k)}}$ are all zero,  we have that $\xi_{c,i}=\xi_{c,k}=\gamma_{c,(i,k)}>2$ and
hence
\[
_3F_2\left(\gamma_{c,(i,k)},1,1;\xi_{c,i},\xi_{c,k};1 \right)={}_2F_1(1,1;\gamma_{c,(i,k)};1)
=\frac{\gamma_{c,(i,k)}-1}{\gamma_{c,(i,k)}-2},
\]
where $1\leq i\neq k\leq n$.
Consequently, we obtain that, for $1\leq i\neq k\leq n$,  correlation (\ref{cov-pos-gen}) reduces to
\[
\mathbf{Corr}[X_i,X_k]=\frac{1}{\gamma_{c,(i,k)}},
\]
which recovers the Pearson correlation of the classical multivariate Pareto distribution of Arnold (1983, 2015).

\begin{corollary}
 \label{cov-cs}
Let $\gamma_{c,(i,k)}\equiv 0$ and assume that the rest of the conditions in Theorem \ref{joint_m} hold, then
 \begin{equation}
 \label{cov-com-gen}
\mathbf{Corr}[X_i,X_k]=\sqrt{\frac{(\xi_{c,i}-2)(\xi_{c,k}-2)}{\xi_{c,i}\xi_{c,k}}}
\frac{\alpha_{c,(i,k)}}{\left(\xi_{c,i,k}-2\right)},
\end{equation}
where $1\leq i\neq k\leq n$.
\end{corollary}
\noindent Formula (\ref{cov-com-gen}) conforms to the one derived in Asimit et al. (2010).

While the importance of the derivations in this section is clear from the distribution theory point of view, some
connections to applications in (actuarial) risk management are perhaps in place.  In summary:
\begin{itemize}
\item The Lebesgue decomposition derived in Theorem \ref{decomp} becomes useful when developing
expressions for, e.g., the joint moments of the MRF distributed random pairs. These moments are
of central importance in (actuarial) risk management, as well as in actuarial and economic
pricing, when the risk measure and pricing functionals belong to the classes of  weighted
(Furman \& Zitikis (2008a,b)) or distorted (Wang (1996)) expectations.
\item Theorem \ref{joint_m} provides an easy to compute expression for the Pearson index of
linear correlation for a random pair coming from the MRF structures. Noticeably, the Pearson correlation can attain
any value in $[0,\ 1]$ and as such Theorem \ref{joint_m} emphasizes that the MRF models cover the entire range of
positive dependence when it is measured by the Pearson $\rho$.
It is well-known that defaults tend to cluster in reality, thus suggesting that they are
positively dependent. Factor models \`a la the ones discussed in Subsection \ref{Subsec_Su}
have been employed extensively to describe the just-mentioned clustering in practice.
The reason stems from a common belief that simultaneous defaults are rare and can be
dismissed. The finding of Das et al. (2007) along with Theorem \ref{joint_m}
suggest that the general MRF dependence structures may provide a better way to model real life default times.
\end{itemize}
\end{section}

\begin{section}{Applications to financial risk measurement}
\label{Sec-5}
In the sequel we consider routes to utilize the results derived hitherto. To this end, we note in passing that
the functional $H:\mathcal{X}\rightarrow [0,\infty]$ is called a risk measure for any risk r.v. $X\in\mathcal{X}$. Moreover,
the aforementioned functional is an actuarial premium calculation principle
(p.c.p.) if the bound $H[X]\geq \mathbf{E}[X]$ holds for all $X\in\mathcal{X}$ with finite
expectations.

Regulatory accords around the world require that insurance companies carry out extensive and quantitatively sound
assessments of their risks, and default risks are not an exception.
In this subsection we report expressions for arguably the most popular risk measure functionals
used nowadays in insurance industry when the risk portfolio is formally described by the MRF structures.

The literature on risk measures is
vast and growing quickly. The following two
indices have however earned an unprecedented amount of interest among both practitioners and theoreticians.
\begin{definition}
\label{VaRDef}
Let $X\in\mathcal{X}$ and fix $q\in[0,\ 1)$, then the Value-at-Risk (VaR) and the Conditional
Tail Expectation (CTE) risk measures are respectively given by
\begin{equation}
\label{VaRdef}
VaR_q[X]=\inf\{x\in\mathbf{R}:\mathbf{P}[X\leq x]\geq q\}
\end{equation}
and
\begin{equation}
\label{CTEdef}
CTE_q[X]=\mathbf{E}[X|\ X>VaR_q[X]].
\end{equation}

\end{definition}

In the context of the MRF risk portfolios, the VaR and the CTE risk measures are readily obtained using
Theorem \ref{marginal-p}.  The proofs are omitted, as they are similar to the proofs of Corollary 3.1 and Corollary 3.3
in Su \& Furman (2016).
\begin{proposition}
Let $\mathbf{X}\sim Pa_{1,\ldots,n}^c(II)(\boldsymbol{\sigma},\ \boldsymbol{\xi})$ as in Definition \ref{MRF-stoch-Def},
then, for $i=1,\ldots,n$,  the Value-at-Risk risk measure is given by
\begin{eqnarray}
\label{var}
VaR_q[X_i]=\sigma_i\left((1-q)^{-1/\xi_{c,i}}-1 \right),
\end{eqnarray}
and the Conditional Tail Expectation risk measure is, for $\xi_{c,i}>1$, given by
\begin{eqnarray}
\label{CTEpar}
CTE_q[X_i]&=&\mathbf{E}[X_i] \frac{\overline{F}_{X_i^\ast}(VaR_q[X_i])}{1-q}+VaR_q[X_i] \nonumber\\
&=&\mathbf{E}[X_i]+VaR_q[X_i]\frac{\xi_{c,i}}{\xi_{c,i}-1},
\end{eqnarray}
where $X_i^\ast \sim Pa(II)(\sigma_i,\xi_{c,i}-1)$.
\end{proposition}

As we have already mentioned, the minima and the maxima r.v.'s, $X_-$ and $X_+$, respectively, play an important role in the theory of credit risk and in insurance mathematics - think, e.g., of the first and last default times - for the former subject, and of
the joint and last-survivor life statuses - for the latter subject.
Next two propositions provide expressions for the
average excess-of-the $q$-th quantile
time of first and last default. The proofs are again similar to these in Su \& Furman (2016) and are thus omitted.

\begin{proposition}
\label{CTEparMin}
In the context of the MRF risk portfolios, the CTE risk measure of the
minima r.v. can be written, if finite and for $q\in[0,\ 1)$, as
\[
CTE_q[X_-]=\mathbf{E}[X_-]\frac{\overline{F}_{X_-^\ast}(VaR_q[X_{-}])}{1-q}+VaR_q[X_{-}],
\]
where $X_-^\ast \sim Pa(\alpha_+(\boldsymbol{\sigma}),\ \xi_{c,1,\ldots,n}+Q-1)$, with
\[\alpha_+(\boldsymbol{\sigma})=\left(
\bigvee_{j=1}^{l} \bigwedge_{i\in\mathcal{RC}_j} \sigma_i
\right)
\vee
\left(
\bigvee_{j=l+1}^{l+m}\left(\sum_{i\in\mathcal{RC}_j}\frac{1}{\sigma_i}\right)^{-1}
\right),\] and $Q$ is an integer-valued and non-negative r.v. with the p.m.f. obtained from the p.m.f.
of $K$ (Lemma \ref{FL2005}) with the help of the following change of measure
\begin{equation}
\label{cmkq}
q_k=\frac{1}{\mathbf{E}[X_{-}]}\frac{\sigma_+}{\xi_{c,1,\ldots,n}+k-1}p_k,\
k=0,1,\ldots
\end{equation}
\end{proposition}
\begin{proposition}
\label{ctemaxima}
In the context of the MRF risk portfolios, the CTE risk measure of the
maxima r.v. can be written, if finite and for $q\in[0,\ 1)$,  as
the following linear combination
\[
CTE_q[X_+]=\frac{1}{1-q}\sum_{\mathcal{S}\subseteq \{1,\ldots,n\}} (-1)^{|\mathcal{S}|-1} \mathbf{E}[X_{\mathcal{S}-}|X_{\mathcal{S}-}>VaR_q(X_+)] \overline{F}_{{\mathcal{S}-}}(VaR_q(X_+)),
\]
where $X_{\mathcal{S}-}=\wedge_{s\in \mathcal{S}\subseteq\{1,\ldots,n\}}X_s$ and
$X_{\mathcal{S}-}\sim F_{\mathcal{S}-}$.
\end{proposition}

We conclude this subsection by studying one more important index that aims at shedding light on the effect of the interdependence between two
random default times. We refer to it as the `solvency bonus' index.
\begin{definition}
\label{bonus-def} Let $X_i$ and $X_k$ denote two possibly correlated random default times of two business units $i$ and $k$, where $1\leq i\neq k\leq n$.
Also, let $y\in\mathbf{R}_+$ be fixed, then
\[
\beta(y;\ X_i,\ X_k)=\mathbf{E}[X_i|\ X_k>y]-\mathbf{E}[X_i]
\]
yields the change in the expected default time of the risk component $i$ given that the risk component $k$ has not defaulted and that $i$ and $k$ are parts of a risk portfolio
rather than in the case when $i$ is a stand alone business unit.
\end{definition}

A simple observation about the solvency bonus index is perhaps appropriate. Namely, if the default times are independent(positively
quadrant dependent), then $\beta(y;\ X_i,\ X_k)=0(\geq 0)$, respectively, for all $y\in\mathbf{R}_+$. The former part of the assertion is
straightforward, whereas the latter part follows because
\[
\mathbf{E}[X_i|\ X_k>y]=\frac{\mathbf{E}[X_i\mathbf{1}\{X_k>y\}]}{\mathbf{P}[X_k>y]}\geq \mathbf{E}[X_i],
\]
for any $y\in\mathbf{R}_+$ (Lehmann (1966)).

\begin{proposition}
\label{survival-function}
Let $i$ and $k$, $1\leq i\neq k\leq n$ be two risk components with respective default
times $X_i$ and $X_k$. In the context of the MRF risk portfolios, the solvency bonus index for the risk component $i$ given that the risk component $k$ has not defaulted is, for $\xi_{c,i}>1$, given by
\begin{eqnarray*}
&&\beta(y;\ X_i,X_k)\\
&=&\frac{\sigma_i}{\xi_{c,i,\overline{(i,k)}}+\gamma_{c,(i,k)}-1}{}_2F_1\left(
\gamma_{c,(i,k)}.1;\xi_{c,i,\overline{(i,k)}}+\gamma_{c,(i,k)},\frac{y/\sigma_k}{1+y/\sigma_k}
\right)
 \\
 &-&
\frac{\sigma_i\left(
1+\frac{y}{\sigma_k}
\right)^{-(\xi_{c,i,\overline{(i,k)}}-\gamma_{c,(i,k)})+1}
\left(
1+2\frac{y}{\sigma_k}
\right)^{-\gamma_{c,(i,k)}}}{ \xi_{c,i,\overline{(i,k)}}+\gamma_{c,(i,k)}-1}
{}_2F_1\left(
\gamma_{c,(i,k)},1;\xi_{c,i,\overline{(i,k)}}+\gamma_{c,(i,k)};\frac{y/\sigma_k}{1+2y/\sigma_k}
\right)
 \\
 &+&
 \frac{\sigma_i}{\xi_{c,i}-1}
 \left(
 1+\frac{y}{\sigma_k}
\right) ^{-(\xi_{c,i,\overline{(i,k)}}-\gamma_{c,(i,k)})+1}
\left(
1+2\frac{y}{\sigma_k}
\right)^{-\gamma_{c,(i,k)}}
\ _2F_1\left(\gamma_{c,(i,k)},1;\xi_{c,i};\frac{y/\sigma_k}{1+2y/\sigma_k} \right)\\
&-&\frac{\sigma_i}{\xi_{c,i}-1},
\end{eqnarray*}
where $y\in\mathbf{R}_{+}$.
\end{proposition}

\end{section}

\section{A numerical example}
\label{Sec-6}

For the sake of the discussion in this section,
we adopt the view of the
Financial Stability Board and the International Monetary Fund
that the systemic risk can be caused by impairment of
all or parts of the financial system, and more formally, we call the risk factor
$j\in\{1,\ldots,l+m\}$ `systemic', if $c_{i,j}=1$ for at least two distinct r.c.'s
$i\in\{1,\ldots,n\}$. Also, we call the risk factor $j\in\{1,\ldots,l+m\}$
`idiosyncratic', if $c_{i,j}=1$ for only one risk component $i\in\{1,\ldots,n\}$.

Consider obligors in a default risk portfolio,
each of
which is exposed to exactly three distinct categories of fatal risk factors, e.g.,
systemic with fully-commonotonic occurrence times of the r.f.'s (category A),
systemic with conditionally independent occurrence times of the r.f.'s (category B)
and idiosyncratic with independent occurrence times of the r.f.'s (category C).
We assume that the
risk factors from distinct risk categories are independent and that
the hitting times (or occurrences) of defaults of the r.c.'s are exponentially-distributed with random parameters distributed gamma.
In fact, the future lifetime r.v. of the $i$-th r.c., $i=1,\ldots,n$ has exponential
distribution with the random parameter $\sigma_i^{-1}\sum_{j=1}^{l+m}c_{i,j}\Lambda_j$, where $\Lambda_j$ are
distributed gamma with unit rate parameters.
Then Definition \ref{MRF-stoch-Def} readily implies that the joint default times of the aforementioned
r.c.'s has d.d.f. (\ref{ddf-gp}).

To illustrate the effect of the dependence structure on the joint default
probability we further set the dimension to $n=2$ and
specialize the set-up above along the lines in Section 16.8 of Engelmann and Rauhmeier (2011)
as well as employing the 2014's Annual Global Corporate Default Study and Rating Transitions of
Standard \& Poor's (Standard \& Poor's (2015)).
More specifically, we assume the existence of six r.f.'s and set each
 $\mu:=\mathbf{E}[\Lambda_j]\equiv 1/1.8$, then fix the time horizon to $15$
years and choose the corresponding default probability, $p$ say, to be equal to $0.3198$
(on par with the `B' credit rating of highly speculative entities).  This yields the multivariate
probability structure of Definition \ref{MRF-stoch-Def} with identically distributed
margins having the parameters $\sigma_i\equiv \sigma=122.39$ and $\xi_{c,i}\equiv 3.33$, for $i=1$ and $2$.

Then we explore three different exposures of the obligors to the systemic and
 idiosyncratic  r.f.'s.
The distinct exposures are stipulated by appropriate choices of the $c$ parameters gathered by matrices $A_c^{(k)},\
k=1,\ 2,\ 3$. We compare the aforementioned three exposures with the reference case
in which no systemic risk presents, that is the joint d.d.f. of default times
is a bivariate Pareto with independent margins.
 We note in passing that the expressions for the d.d.f.'s below readily follow
from Theorem \ref{def-gp}, whereas the values of the Pearson correlation coefficient
 are in non-trivial cases obtained  with the help of Theorem \ref{joint_m}.

\begin{itemize}
\item [Case (1).]
Only the systemic (category A) and idiosyncratic (category C) risks present. The exposure
is represented schematically with the use of the following matrix, in which the rows
and the columns represent r.c.'s and r.f.'s, respectively
\[
A_c^{(1)}=\left(
\begin{array}{cccc||cccc}
1 & 1 & 1 & 1 & 1 & 1 &0 & 0\\

1 & 1 & 1 & 1 & 0 & 0& 1 & 1\\
\end{array}
\right).
\]
The joint d.d.f. of the risk components is given by
\[
\overline{F}^{(1)}(x_1,x_2)=\left(
1+\frac{\max(x_1,x_2)}{\sigma}
\right)^{-4\mu} \left(
1+\frac{x_1}{\sigma}
\right)^{-2\mu}\left(
1+\frac{x_2}{\sigma}
\right)^{-2\mu}  ,
\]
where $x_1,x_2$ are all in $\mathbf{R}_{+}$.
This is obviously the d.d.f. of the bivariate Pareto distribution of Asimit et al. (2010).
In this r.p., the Pearson correlation coefficient between the r.c.'s is $0.36$.
\end{itemize}
\noindent

\begin{itemize}
\item [Case (2).]
There are four conditionally independent r.f.'s (category B) and two uncorrelated idiosyncratic r.f.'s (category C).  The exposure is gathered by the following block matrix
\[
A_c^{(2)}=\left(
\begin{array}{c|cccc|cccc}
~ & 1 & 1 & 1 & 1 & 1 & 1 &0 & 0\\

~ & 1 & 1 & 1 & 1 & 0 & 0& 1 & 1\\
\end{array}
\right).
\]
The joint d.d.f. of the risk components is given by
\begin{eqnarray*}
&&\overline{F}^{(2)}(x_1,x_2)\\
&&=\left(
1+\frac{x_1+x_2}{\sigma}
\right)^{-4\mu}
\left(
1+\frac{x_1}{\sigma}
\right)^{-2\mu}
\left(
1+\frac{x_2}{\sigma}
\right)^{-2\mu},
\end{eqnarray*}
where $x_1,x_2$ are all in $\mathbf{R}_{+}$. This case corresponds to the bivariate Pareto model of Su \& Furman (2016).
In this r.p., the Pearson correlation coefficient between the r.c.'s is $0.14$.

\item [Case (3).] The r.p. admits the most general form that is proposed in the current paper. Namely r.f.'s from all three categories (A, B and C) present.
The exposure block matrix is given by
\[
A_c^{(3)}=\left(
\begin{array}{cc|cc|ccccc}
1 & 1 & 1 & 1 & 1 & 1 &0 & 0\\

1 & 1 & 1 & 1 & 0 & 0& 1 & 1\\
\end{array}
\right).
\]
The joint d.d.f. of the risk components is
\[
\overline{F}^{(3)}(x_1,x_2)=\left(
1+\frac{\max(x_1,x_2)}{\sigma}\right)^{-2\mu}
\left(
1+\frac{x_1+x_2}{\sigma}\right)^{-2\mu}
\left(
1+\frac{x_1}{\sigma}
\right)^{-2\mu},
\left(
1+\frac{x_2}{\sigma}
\right)^{-2\mu},
\]
where $x_1,x_2$ are all in $\mathbf{R}_{+}$.
In this r.p., the Pearson correlation coefficient between the r.c.'s is equal to $0.23$.
\end{itemize}

\subsection{Expected times of the first default} The left panel of
Figure \ref{fig:CTEmin}
depicts the values of $CTE_q[X_-]$ for $q\in[0,\ 1)$, $X_-\in\mathcal{X}$
and portfolios (1) to (3) as well as the reference portfolio, denoted by $(\perp)$.
As the risk components are identically distributed, it is not difficult to see that
the following ordering holds
\begin{equation}
\label{stord1}
\overline{F}^{(1)}_-\geq_{st}\overline{F}^{(3)}_-\geq_{st}\overline{F}^{(2)}_-\geq_{st}
\overline{F}^{(\perp)}_-,
\end{equation}
where `$\geq_{st}$' denotes first order stochastic dominance (FSD). Furthermore,
since the CTE risk measure
is known to preserve the FSD ordering, we also have that
\[
CTE^{(1)}_q[X_-]\geq CTE^{(3)}_q[X_-]\geq CTE^{(2)}_q[X_-]\geq CTE^{(\perp)}_q[X_-]
\]
for all $q\in[0,\ 1)$ and $X_-\in\mathcal{X}$. This conforms to
Figure \ref{fig:CTEmin} (left panel), which hints that the r.p.'s with more
significantly correlated r.c.'s
enjoy higher, and thus more favourable, occurrence times of the first default.
\begin{figure}[h!]
\centering
\includegraphics[width=7cm, height=7cm]{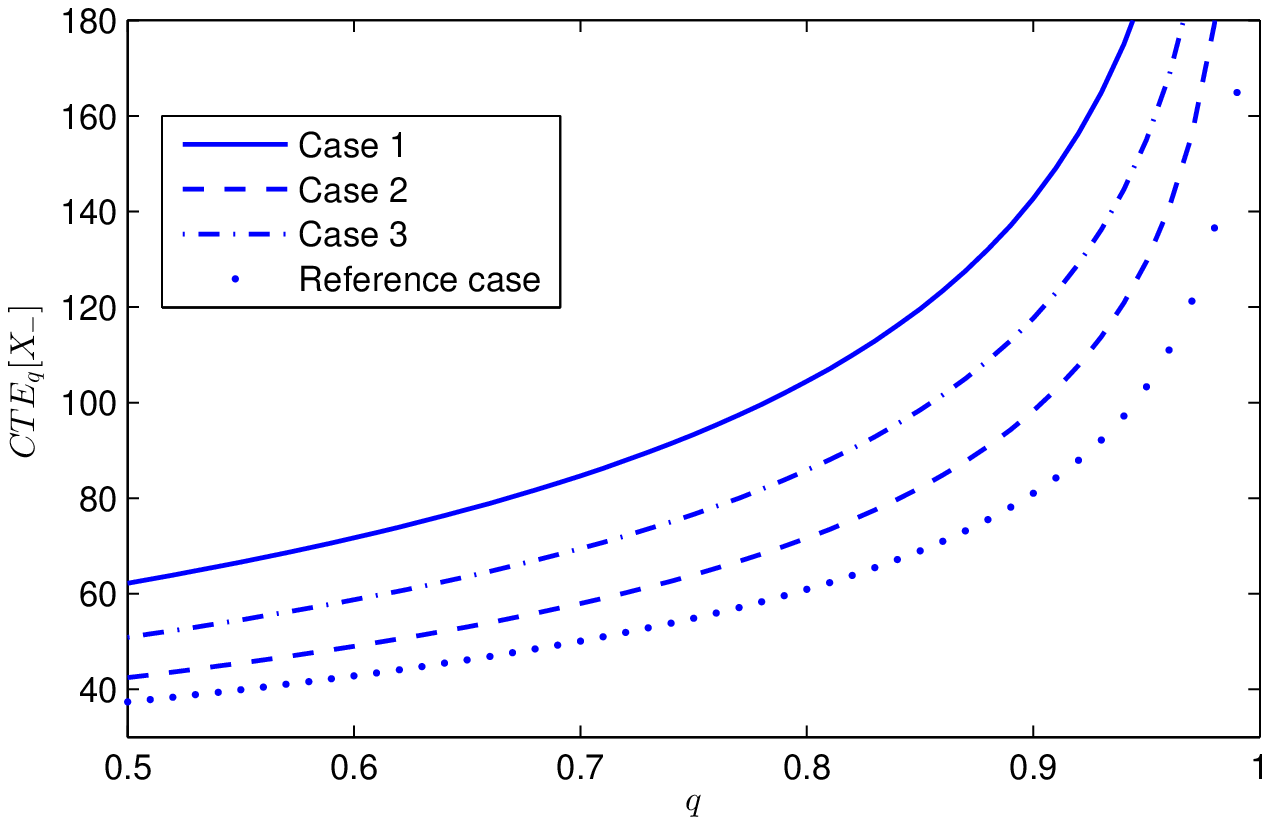}
 \includegraphics[width=7cm, height=7cm]{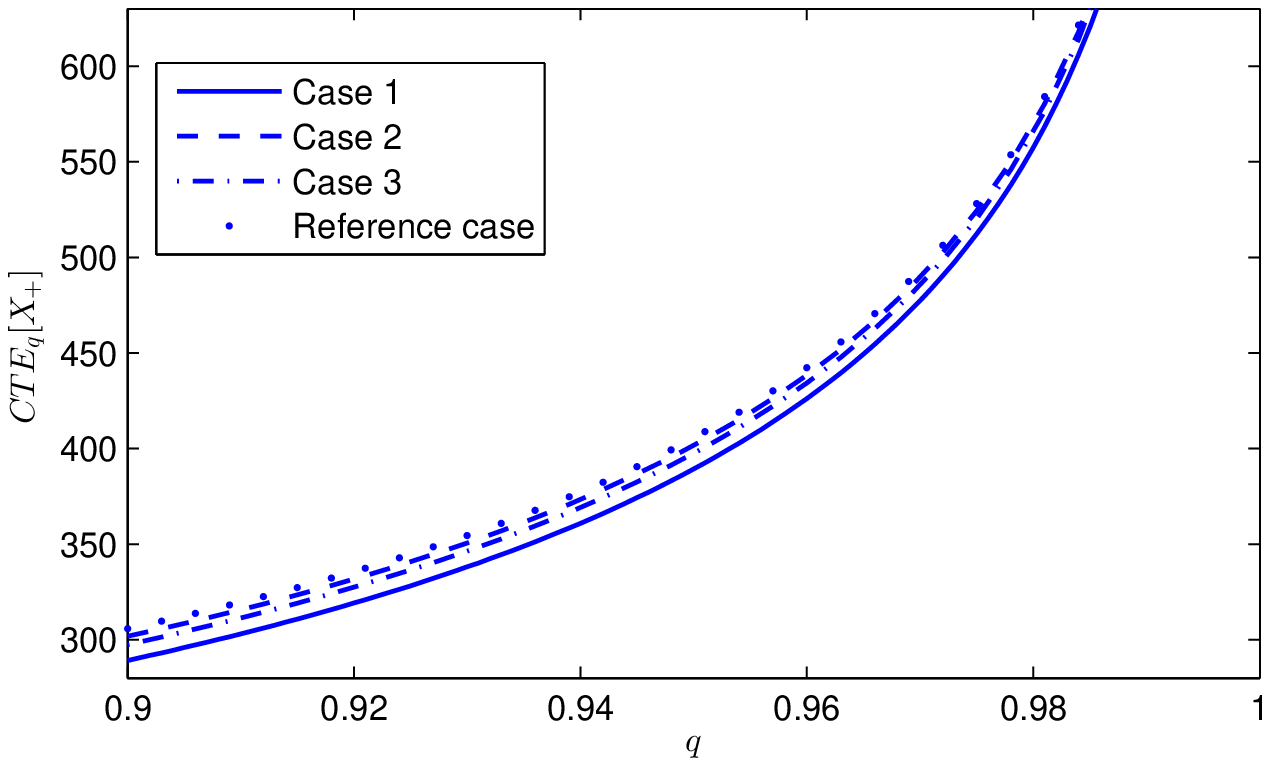}\\
\caption{Conditional expected times of first (left panel) and last
(right panel) default for portfolios (1) - (3) and the
reference portfolio $(\perp)$ for `B' rating r.p.'s with the probability of
default $p=0.3198$ and
$\mu=1/1.8$.
Proposition \ref{CTEparMin} is employed to
 compute the values of $CTE_q$ for $q\in[0,\ 1)$.}
\label{fig:CTEmin}
\end{figure}

The
downside of high correlations is elucidated in
 Figure \ref{fig:minCase}, in which
we leave the probability of default $p$ to be equal to $0.3198$ (`B' rating),
but vary the $\mu$ parameter that stipulates the effect of the risk factors.
In this respect, we observe that the r.p.'s with stronger correlations between r.c.'s are more
sensitive to the changes in the $\mu$ parameter, and therefore such r.p.'s
must be monitored and stress-tested more frequently.

\begin{figure}[h!]
\centering
\includegraphics[width=7cm,height=7cm]{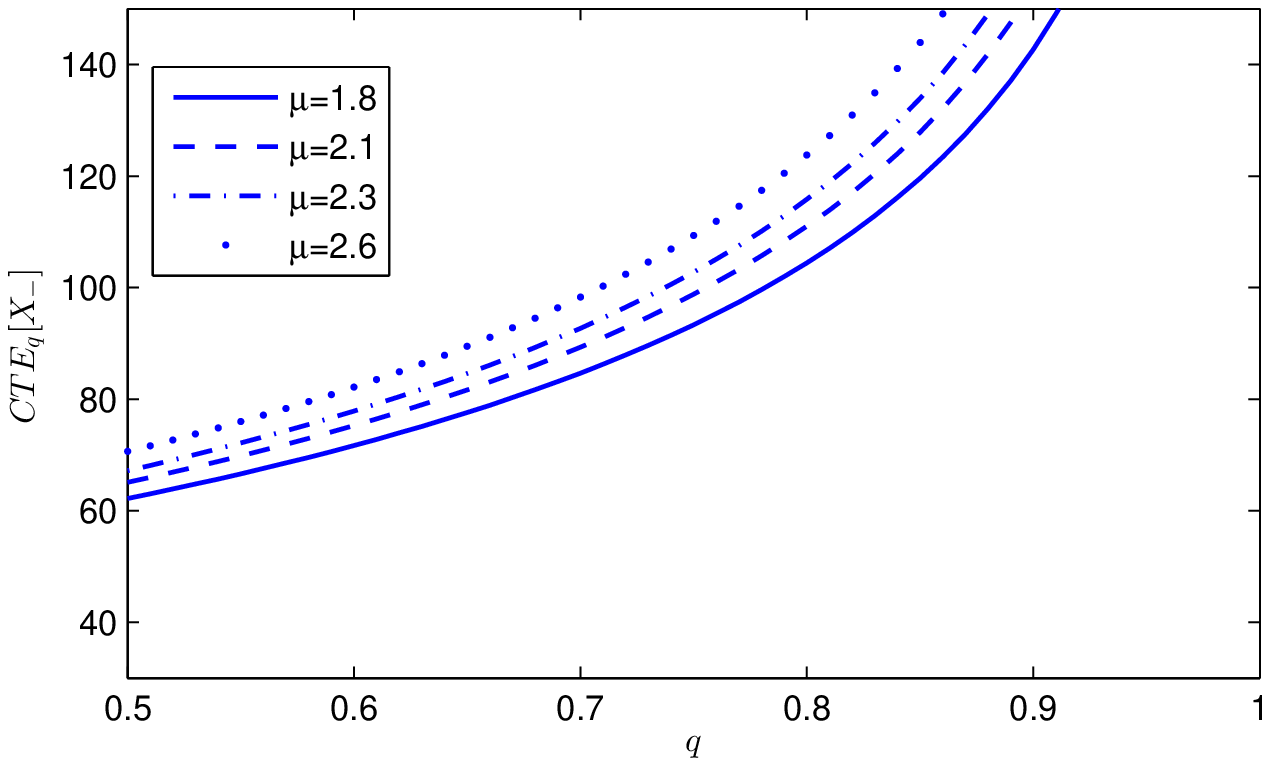}
\includegraphics[width=7cm,height=7cm]{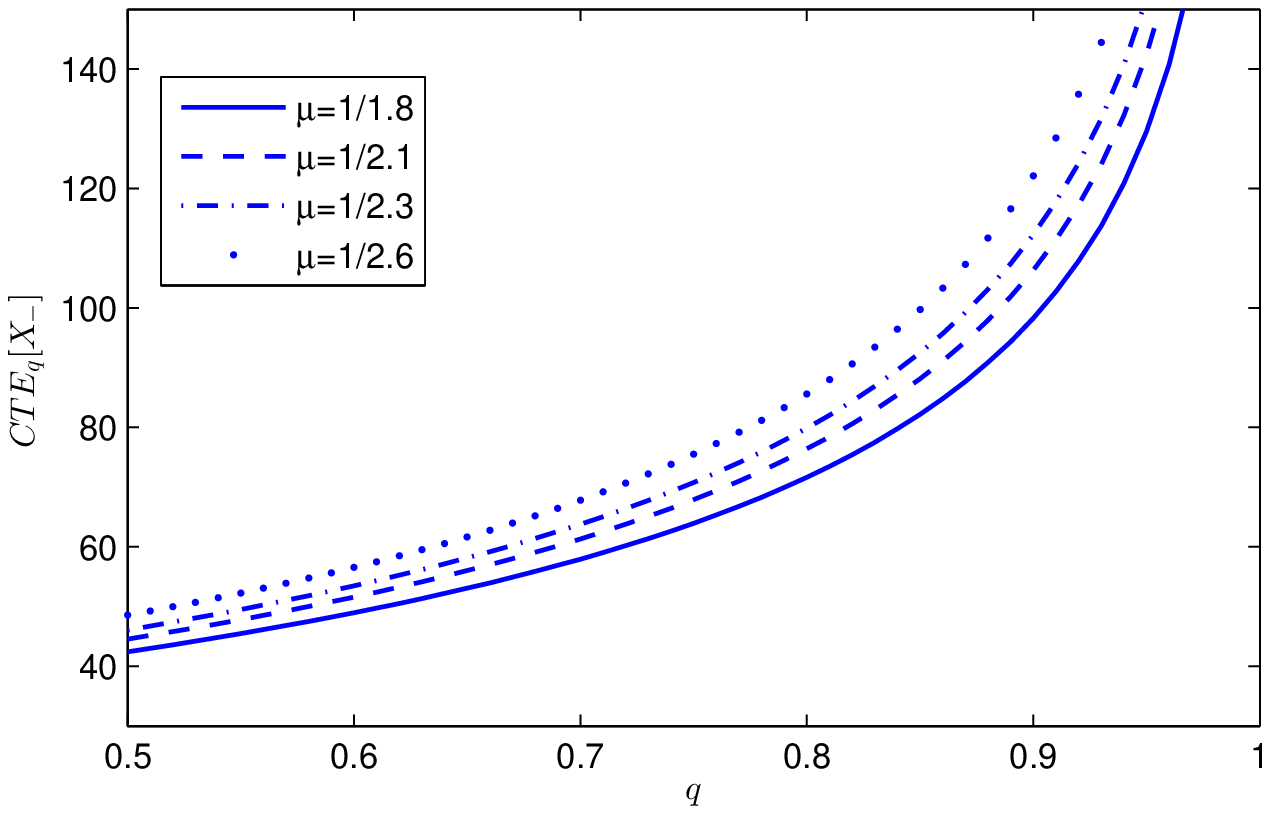}\\
\includegraphics[width=7cm,height=7cm]{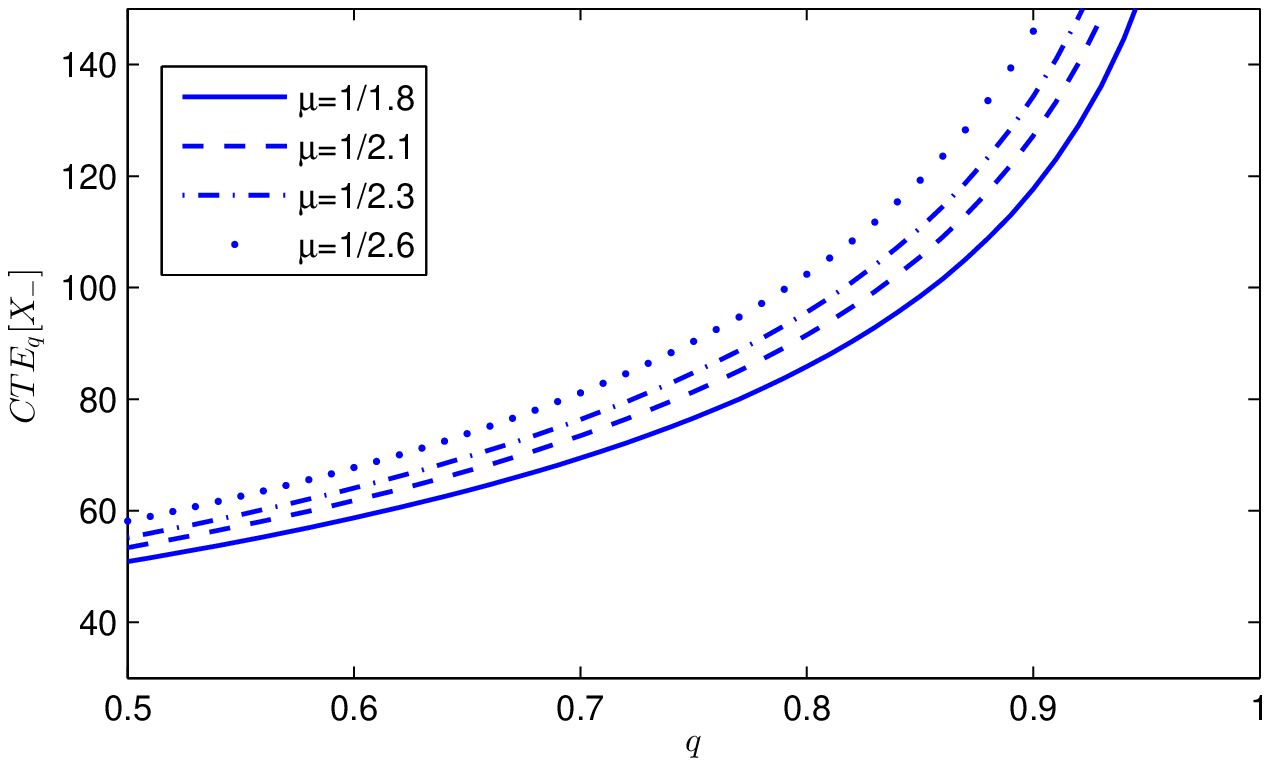}
\includegraphics[width=7cm,height=7cm]{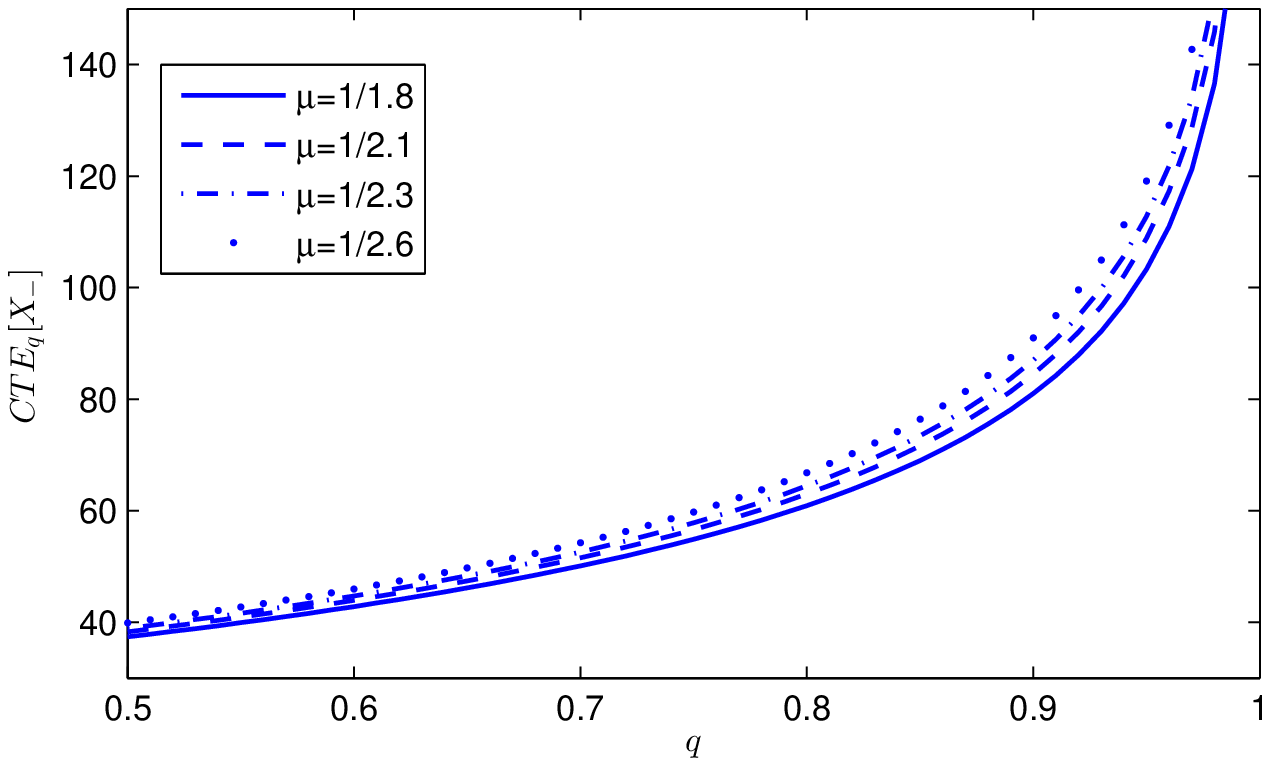}
\caption{Conditional expected times of first default for portfolios (1) (top left panel),
(2) (top right panel), (3) (bottom left panel) and reference $(\perp)$ (bottom right panel)
with the parameter $\mu$ varying from $1/1.8$ to $1/2.6$ and the default probability
$p=0.3198$. Proposition \ref{CTEparMin} is employed to compute the values of
$CTE_q$ for $q\in[0,\ 1)$.}
\label{fig:minCase}
\end{figure}

\subsection{Expected times of the last default}
Figure \ref{fig:CTEmin} (right panel)
depicts the values of $CTE_q[X_+]$ for $q\in[0,\ 1),$ $X_+\in\mathcal{X}$ and portfolios
(1) to (3) as well as the reference portfolio $(\perp)$.
We have that
\[
\overline{F}^{(\perp)}_+\geq_{st}\overline{F}^{(2)}_+\geq_{st} \overline{F}^{(3)}_+
\geq_{st} \overline{F}^{(1)}_+,
\]
and hence
\[
CTE^{(\perp)}_q[X_+]\geq CTE^{(2)}_q[X_+]\geq CTE^{(3)}_q[X_+] \geq CTE^{(1)}_q[X_+]
\]
for all $q\in[0,\ 1)$ and $X_+\in\mathcal{X}$.
This conforms with the right panel of Figure \ref{fig:CTEmin}.

Unlike in the case of the first default,
we observe that if the time of the last default is
of interest and  the distributions of the r.c.'s are fixed,
then assuming stronger correlations between r.c.'s yields a more
conservative assessment of the expected time of the last default.

\subsection{Solvency bonus indices}
Figure \ref{fig:econCTE} depicts the solvency bonus indices for portfolios (1) to (3). As expected, stronger dependencies between default times are associated with higher values of $\beta$.

\begin{figure}[h!]
\centering
\includegraphics[width=7cm,height=7cm]{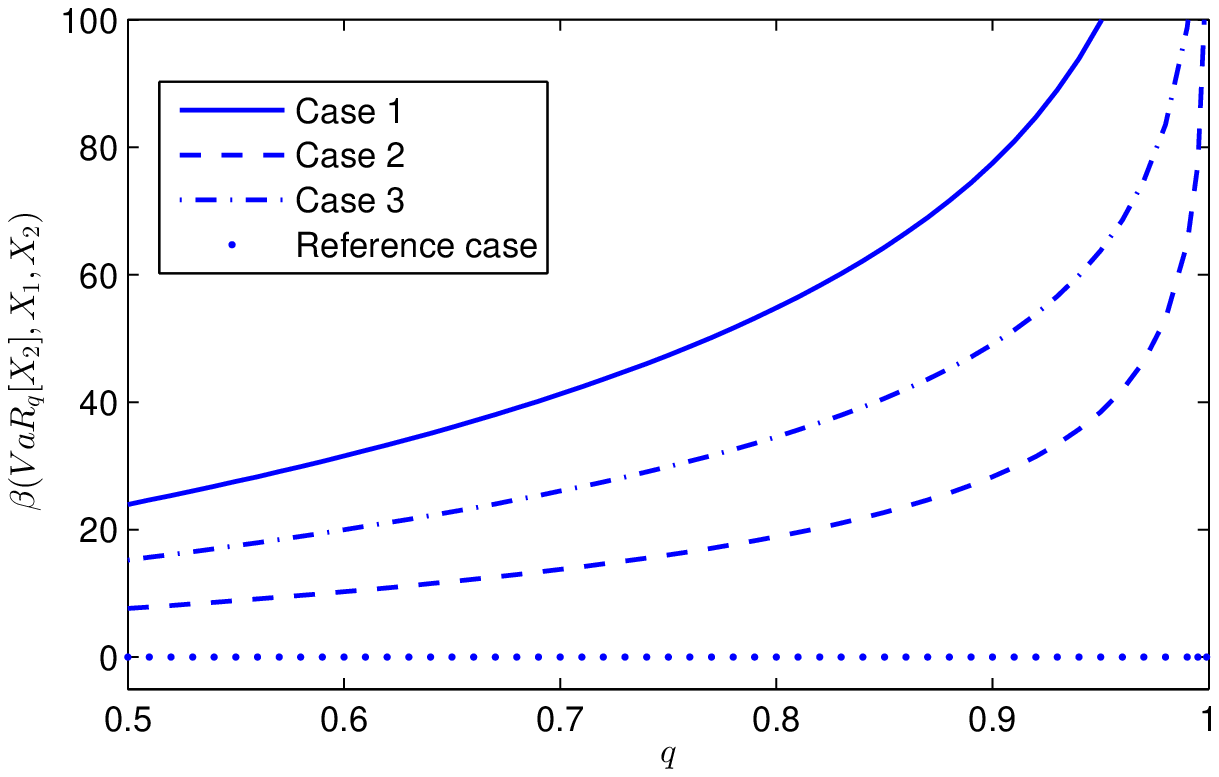}
\includegraphics[width=7cm,height=7cm]{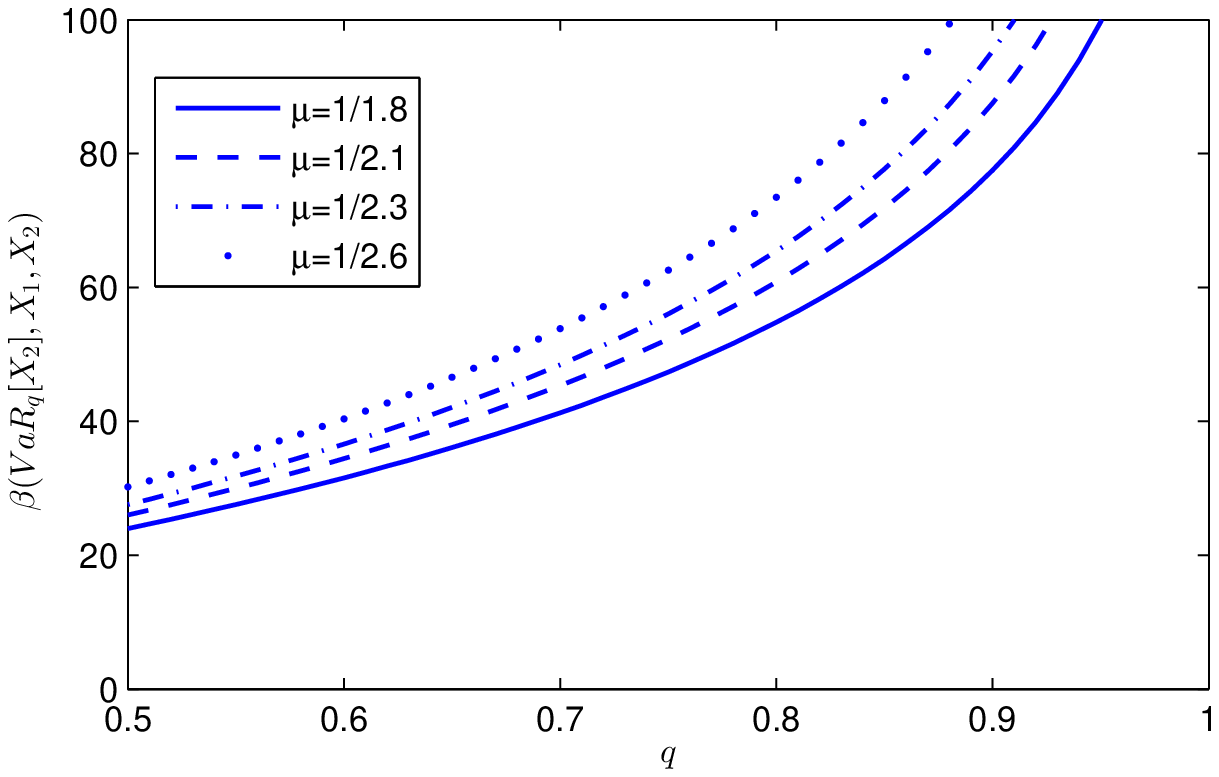}\\
\includegraphics[width=7cm,height=7cm]{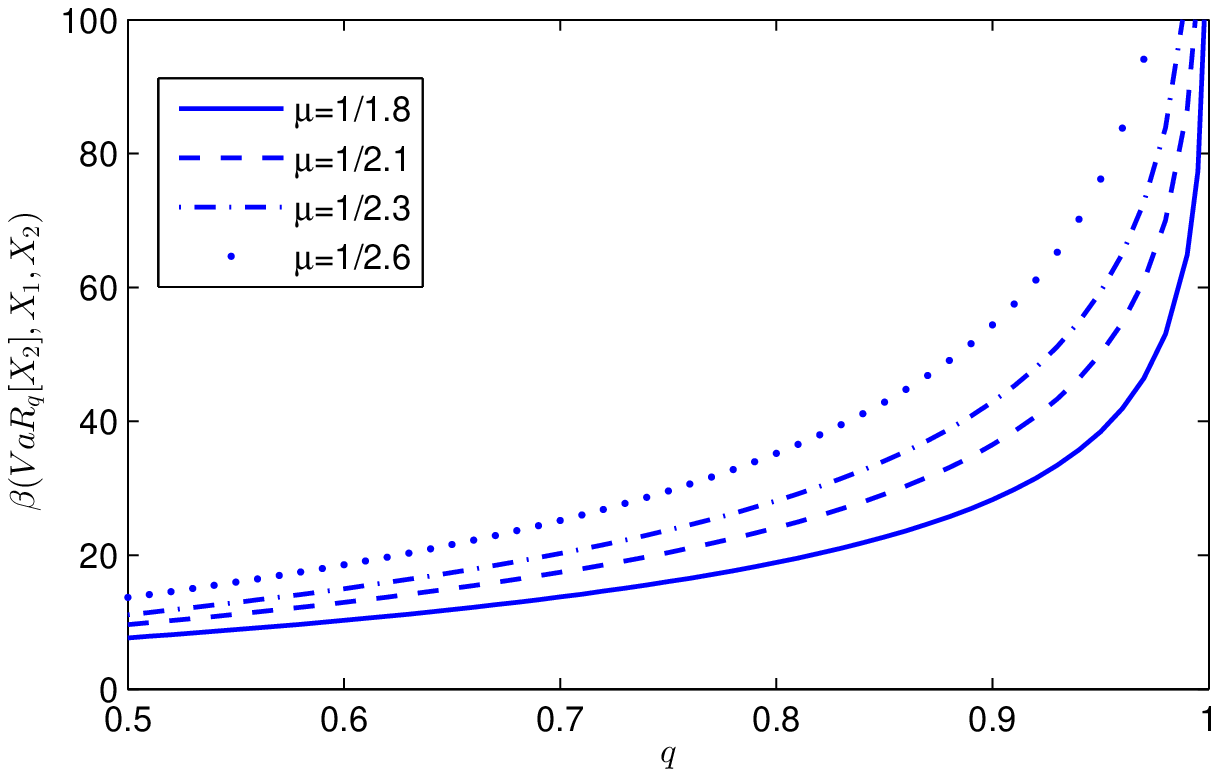}
\includegraphics[width=7cm,height=7cm]{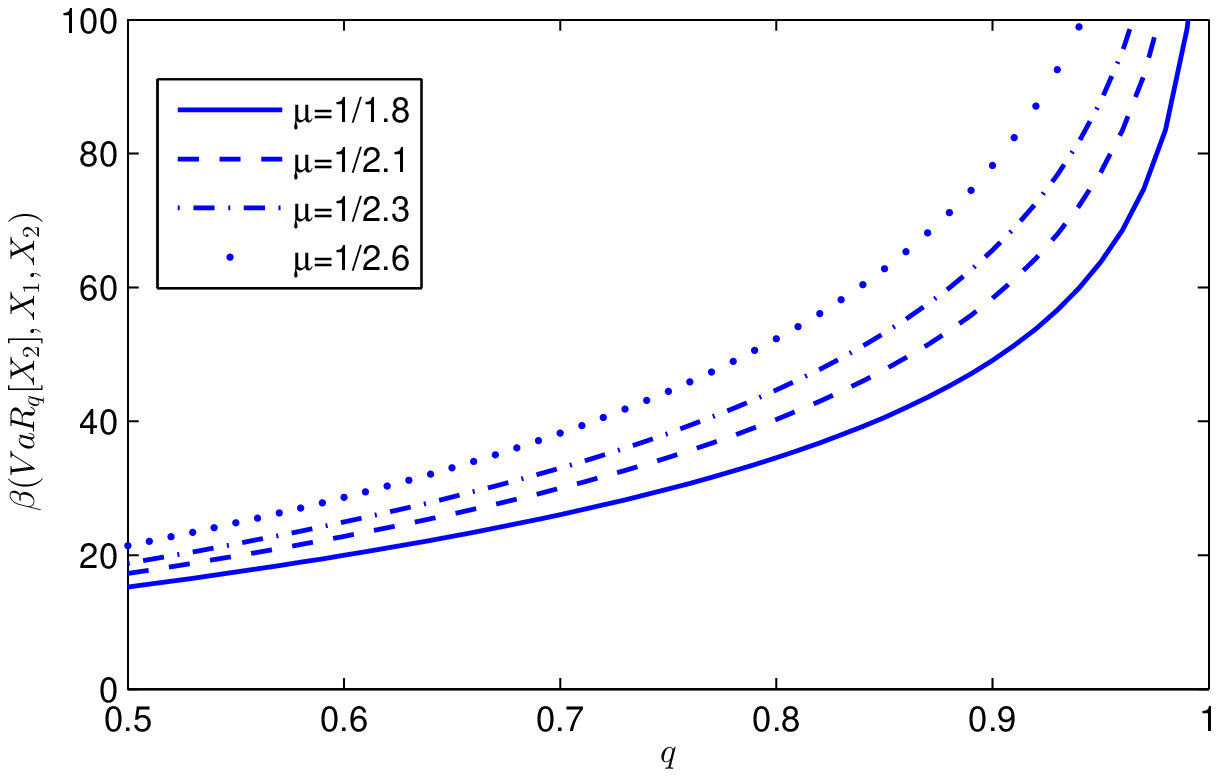}
\caption{Comparison of solvency bonus indices for portfolios (1) to (3) (top left panel), as well as the solvency bonus indices solely for portfolio (1) (top right panel), (2) (bottom left panel) and (3) (bottom right panel)
with the parameter $\mu$ varying from $1/1.8$ to $1/2.6$ and the default probability
$p=0.3198$. Proposition \ref{survival-function} is used to compute the values of the solvency bonus indices.}
\label{fig:econCTE}
\end{figure}

\section{Conclusions}
\label{Sec-7}
The latest Solvency II directives require the insurers to recognize the interdependencies within and among different liability classes.
As such, the new paradigm arguably
brings to an end the indisputable role of the assumption of independence that has shaped both research  and applications in actuarial science
and general quantitative risk management in the 20th century. The choice of an appropriate probability dependence model is however not an easy call.
Indeed, while there exists only one way to describe stochastic independence, the forms of stochastic dependence are infinite.

In this paper we have introduced  a new class of Multiple Risk Factor dependence structures. On the one hand, these structures
emerge as an extension of the
popular CreditRisk$^+$ approach, and as such they formally describe default risk portfolios exposed to an arbitrary number of
fatal risk factors with conditionally exponential hitting times that can be independent, positively dependent and even fully comonotonic.
On the other hand, the MRF structures can be viewed as a quite general family of multivariate probability distributions with Pareto of the 2nd kind
univariate margins, and in this role, they can model risk portfolios of (insurance) losses with heavy tailed and positively dependent
risk components.

It often happens in mathematical sciences that generalizing an object highlights the essence of the matter and helps to understand it better.
By generalizing the classical multivariate Pareto distribution of Arnold (1983, 2015), among others, in this paper we have extended the range of the attainable Pearson correlations to the entire $[0,\ 1]$ interval
and complemented some formal dependence analysis. As factor models have been known to produce under-correlated default times,
the MRF structures introduced and studied herein may provide a possible route to `inject the required amounts of extra correlation',
and they may thereby become of interest to banks, credit unions and insurance companies.
In addition, as the new realms of excessively prudent risk management make particular
effort to model non-hedgeable and heavy tailed risks, the family of MRF multivariate
distributions may be of interest well beyond the context of credit risk.

\section*{Acknowledgements}
We are grateful to Prof. Dr. Paul Embrechts and all participants of the ETHs
Series of Talks in Financial and Insurance Mathematics for feedback and insights.

\section*{Disclosure statement}
No potential conflict of interest was reported by the authors.

\section*{Funding}
Our research has been supported by the Natural Sciences and Engineering Research Council (NSERC) of Canada. Jianxi Su also acknowledges the financial support of the Government of Ontario and MITACS Canada via, respectively, the Ontario Graduate Scholarship program
and the Elevate Postdoctoral fellowship.

\section*{References}
\hangindent=\parindent\noindent {Adalsteinsson, G. (2014).} \textit{The Liquidity Risk Management Guide: From Policy to Pitfalls.} Wiley, Chichester.

\hangindent=\parindent\noindent {Arnold, B. C. (1983).} \textit{Pareto Distributions.} International Cooperative Publishing House, Fairland.

\hangindent=\parindent\noindent {Arnold, B. C. (2015).} \textit{Pareto Distributions}, 2nd ed. CRC Press, Boca Raton.

\hangindent=\parindent\noindent Asimit, A. V., Furman, E. \& Vernic, R. (2010). On a multivariate Pareto distribution. \textit{Insurance: Mathematics and Economics} \textbf{46}(2), 308--316.

\hangindent=\parindent\noindent Asimit, A. V., Furman, E. \& Vernic, R. (2016). Statistical inference for a new class of multivariate Pareto distributions. \textit{Communications in Statistics - Simulation and Computation} \textbf{45}(2),
456--471.

\hangindent=\parindent\noindent Azizpour, S. \& Giesecke, K. (2008). {Self-exciting corporate defaults : Contagion vs. frailty.} Technical report, Stanford University, Stanford.

\hangindent=\parindent\noindent Benson, D. A., Schumer, R. \& Meerschaert, M. M. (2007). Recurrence of extreme events with power-law interarrival times. \textit{Geophysical Research Letters} \textbf{34}, 1--5.

\hangindent=\parindent\noindent Bowers, N. L., Gerber, H. U., Hickman, J. C., Jones, D. A. \& Nesbitt, C. J. (1997). \textit{Actuarial Mathematics}, 2nd ed. Society of Actuaries, Schaumburg.



\hangindent=\parindent\noindent Bunke, H. C. (1969). \textit{A Primer on American Economic History.} Random House, New York.

\hangindent=\parindent\noindent Cebri\'{a}n, A. C., Denuit, M. \& Lambert, P. (2003). Generalized Pareto fit to the Society of Actuaries' large claims database. \textit{North American Actuarial Journal} \textbf{7}(3), 18--36.

\hangindent=\parindent\noindent Chavez-Demoulin, V., Embrechts, P. \& Hofert, M. (2015). An extreme value approach for modeling operational risk losses depending on covariates. \textit{Journal of Risk and Insurance}, in press.

\hangindent=\parindent\noindent Chiragiev, A. \& Landsman, Z. (2009). Multivariate flexible Pareto model: Dependency structure, properties and characterizations. \textit{Statistics and Probability Letters} \textbf{79}(16), 1733--1743.

\hangindent=\parindent\noindent Das, S. R., Duffie, D., Kapadia, N. \& Saita, L. (2007). Common failings: How corporate defaults are correlated. \textit{Journal of Finance} \textbf{62}(1), 93--117.

\hangindent=\parindent\noindent Denuit, M., Dhaene, J., Goovaerts, M. \& Kaas, R. (2005). \textit{Actuarial Theory for Dependent Risk: Measures, Orders and Models.} Wiley, Chichester.


\hangindent=\parindent\noindent Duffie, D. \& Singleton, K. J. (1999). Modeling term structures of defaultable bonds. \textit{Review of Financial Studies} \textbf{12}(4), 687--720.

\hangindent=\parindent\noindent Engelmann, B. \& Rauhmeier, R. (2011). \textit{The Basel II Risk Parameters: Estimation, Validation, Stress Testing - with Applications to Loan Risk Management.} Springer, Berlin.

\hangindent=\parindent\noindent Feller, W. (1966). \textit{An Introduction to Probability Theory and Its Applications.} Wiley, New York.


\hangindent=\parindent\noindent Frey, R. \& McNeil, A. J. (2003). Dependent defaults in models of portfolio credit risk. \textit{Journal of Risk} \textbf{6}(1), 59--92.

\hangindent=\parindent\noindent Furman, E. \& Landsman, Z. (2005). Risk capital decomposition for a multivariate dependent gamma portfolio. \textit{Insurance: Mathematics and Economics} \textbf{37}(3), 635--649.

\hangindent=\parindent\noindent Furman, E. \& Zitikis, R. (2008a). Weighted premium calculation principles. \textit{Insurance: Mathematics and Economics} \textbf{42}(1), 459--465.

\hangindent=\parindent\noindent Furman, E. \& Zitikis, R. (2008b). Weighted risk capital allocations. \textit{Insurance: Mathematics and Economics} \textbf{43}(2), 263--269.


\hangindent=\parindent\noindent Gabaix, X., Gopikrishnan, P., Plerou, V. \& Stanley, H. E. (2003). A theory of power-law
distributions in financial market fluctuations. \textit{Nature} \textbf{423}, 267--270.

\hangindent=\parindent\noindent Geweke, J.  \& Amisano, G. (2011). Hierarchical Markov normal mixture models with applications to financial asset returns. \textit{Journal of Applied Econometrics} \textbf{26}(1), 1--29.

\hangindent=\parindent\noindent Giesecke, K. (2003). A simple exponential model for dependent defaults. \textit{The Journal of
Fixed Income} \textbf{13}(3), 74--83.

\hangindent=\parindent\noindent Gordy, M. B. (2000). A comparative anatomy of credit risk models. \textit{Journal of Banking and Finance} \textbf{24}(1), 119--149.

\hangindent=\parindent\noindent Gradshteyn, I. S.  \& Ryzhik, I. M. (2014). \textit{Table of Integrals, Series, and Products}, 8th ed.
Academic Press, New York.

\hangindent=\parindent\noindent Joe, H. (1997). \textit{Multivariate Models and Dependence Concepts.} CRC Press, Boca Raton.

\hangindent=\parindent\noindent Koedijk, K. G., Schafgans, M. M. A. \& de Vries, C. G. (1990). The tail index of
exchange rate returns. \textit{Journal of International Economics} \textbf{29}(1-2), 93--108.

\hangindent=\parindent\noindent Lehmann, E. L. (1966). Some concepts of dependence. \textit{The Annals of Mathematical Statistics} \textbf{37}(5), 1137--1153.

\hangindent=\parindent\noindent Longin, F. M. (1996). The asymptotic distribution of extreme stock market returns.
\textit{Journal of Business} \textbf{69}(3), 383--408.

\hangindent=\parindent\noindent Marshall, A. W.  \& Olkin, I. (1967). A generalized bivariate exponential distribution.
\textit{Journal of Applied Probability} \textbf{4}(2), 291--302.

\hangindent=\parindent\noindent McNeil, A. J., Frey, R. \& Embrechts, P. (2005). \textit{Quantitative Risk Management: Concepts, Techniques, and Tools.} Princeton University Press, Princeton.

\hangindent=\parindent\noindent Moschopoulos, P. G. (1985). The distribution of the sum of independent gamma random variables. \textit{Annals of the Institute of Statistical Mathematics} \textbf{37}(1), 541--544.

\hangindent=\parindent\noindent Nelsen, R. B. (2006). \textit{An Introduction to Copulas}, 2nd ed. Springer, New York.

\hangindent=\parindent\noindent Soprano, A., Crielaard, B., Piacenza, F. \& Ruspantini, D. (2010). \textit{Measuring Operational and Reputational Risk: A Practitioner's Approach.} Wiley, Chichester.

\hangindent=\parindent\noindent Standard \& Poor's (2015). Default, transition and recovery: 2014 annual global corporate default study and rating transitions.  Technical report, Standard \& Poor's, New York.

\hangindent=\parindent\noindent Su, J. \& Furman, E. (2016). A form of multivariate Pareto distribution with applications to financial risk measurement. \textit{ASTIN Bulletin: The Journal of the International Actuarial Association}, in press.

\hangindent=\parindent\noindent Sweeting, P. (2011). \textit{Financial Enterprise Risk Management.} Cambridge University Press, Cambridge.

\hangindent=\parindent\noindent Vernic, R. (2011). Tail conditional expectation for the multivariate Pareto distribution of the second kind: Another approach. \textit{Methodology and Computing in Applied Probability} \textbf{13}(1), 121--137.

\hangindent=\parindent\noindent Wang, S. (1996). Premium calculation by transforming the layer premium density. \textit{ASTIN Bulletin: The Journal of the International Actuarial Association} \textbf{26}(1), 71--92.

\appendix
\numberwithin{equation}{section}
\section{Proofs}
\label{app-proof}
\begin{proof}[Proof of Lemma \ref{FL2005}]
The first part of the lemma, i.e., the distribution of $\Lambda^\ast$ is proved in Moschopoulos (1985) (also, Furman \& Landsman (2005)). To
establish the second part, we note that, for $x\in\mathbf{R}_+$, the d.d.f. of $E_{\Lambda^\ast}$ is given by
\begin{eqnarray*}
\overline{F}_{\Lambda^\ast}(x)&=&\int_{\mathbf{R}_+} e^{-x t} \sum_{k=0}^\infty \frac{p_k}{\Gamma(\gamma+k)}e^{-\alpha_+ t}t^{\gamma+k-1}\alpha_+^{\gamma+k}dt =\sum_{k=0}^\infty p_k\left(
1+\frac{x}{\alpha_+}
\right)^{-(\gamma+k)},
\end{eqnarray*}
where the interchange of the summation and the integration signs is justified because of the uniform convergence of the integrand. This completes the
proof.
\end{proof}

\begin{proof}[Proof of Theorem \ref{pro-min}]
We observe that, for $x\in \mathbf{R}_+$, the d.d.f. of the  first default r.v. is
written as
\begin{eqnarray*}
\mathbf{P}[X_-> x]&=&\mathbf{P}[\wedge_{i=1}^n X_i> x]\\
&=&\prod_{j=1}^{l}\left(1+x\bigvee_{i\in\mathcal{RC}_j} \frac{1}{\sigma_i}\right)^{-\xi_j}
\prod_{j=l+1}^{l+m}\left(1+x\sum_{i\in\mathcal{RC}_j}\frac{1}{\sigma_i}\right)^{-\xi_j} \\
&=&\int_{\mathbf{R}_+^{l+m}} \exp\left\{
-x
\sum_{j=1}^{l+m} \lambda_j
\right\}dG_{\boldsymbol{M}}(\lambda_1,\ldots,\lambda_{l+m})
\end{eqnarray*}
where $\boldsymbol{M}\sim G_{1,\ldots,l+m}$ is an $(l+m)$-dimensional r.v. with  stochastically independent and gamma distributed
coordinates. More specifically, we have that $M_j\sim Ga(\xi_j,\ \wedge_{i\in\mathcal{RC}_j}
\sigma_i)$ for $j=1,\ldots,l$ and
$M_j\sim Ga\left(\xi_j,\ \left(\sum_{i\in\mathcal{RC}_j}
1/{\sigma_i}\right)^{-1}\right)$ for $j=l+1,\ldots,l+m$.

Clearly we also have that $\mathbf{P}[X_-> x]=\overline{F}_{E_{M^\ast}}(x),\ x\in\mathbf{R}_+$, where $M^\ast$ is an $(l+m)$- fold convolution
of the coordinates of $\boldsymbol{M}=(M_1,\ldots,M_{l+m})'$, that is its distribution follows from the first part of Lemma
\ref{FL2005}. The statement of the theorem then follows by evoking the second part of Lemma \ref{FL2005}. This completes the proof.
\end{proof}

\begin{proof}[Proof of Corollary \ref{max-th}]
Clearly $\overline{F}_+(x)=\mathbf{P}[\left\{
\cup_{i=1}^n (X_i>x)
\right\}],\ x\in\mathbf{R}_+$, and then the statement of the corollary follows using the inclusion-exclusion principle and
employing Theorem \ref{pro-min}. This completes the proof.
\end{proof}

\begin{proof}[Proof of Theorem \ref{sim-default}]
First, since $\mathbf{P}[Y_{i_1}=\cdots =Y_{i_k}]\equiv0$, we arrive at
\begin{eqnarray}
&&\mathbf{P}[X_{i_1}/\sigma_{i_1}=\cdots=X_{i_k}/\sigma_{i_k}] \notag \\
&=&\mathbf{P}\left[\left\{ \bigcap_{i_h=i_1}^{i_k} (Z_{i_h}\leq Y_{i_h}) \cap \left(Z_{i_1}/\sigma_{i_1}=\cdots
=Z_{i_k}/\sigma_{i_k}\right)\right\}\right] \notag \\
&=&\alpha_{c,(i_1,\ldots,i_k)} \int_{\mathbf{R}_+} \prod_{j=l+1}^{l+m}	 \left(
1+z\sum_{i_h=1}^{i_k} c_{i_h,j}
\right)^{-\gamma_{j}} (1+z)^{-\alpha_{c,i_1,\ldots,i_k}-1}
 dz  \label{ThSingEq1}
\end{eqnarray}
where we have the latter equality sign since, for
$Z_{{i_h,\overline{(i_1,\ldots,i_k)}}}^\bullet:=\bigwedge_{j\in\mathcal{RF}_{{i_h,\overline{(i_1,\ldots,i_k)}}}^l} E_{\Lambda_j},\ i_h=i_1,\ldots,i_k$ and
$Z_{(i_1.\ldots,i_k)}^\bullet:=\bigwedge_{j\in\mathcal{RF}_{(i_1,\ldots,i_k)}^l} E_{\Lambda_j}$,  the following string of expressions follows
\begin{eqnarray*}
&&\mathbf{P}\left[
\frac{Z_{i_1}}{\sigma_{i_1}}=\cdots = \frac{Z_{i_k}}{\sigma_{i_k}}>z\right] \\
&=&\mathbf{P}\left[\left\{
Z_{{i_h,\overline{(i_1,\ldots,i_k)}}}^\bullet > Z_{(i_1,\ldots,i_k)}^\bullet \textnormal{ for all }i_h\in\{i_1,\ldots,i_k\}
\right\}
\cap
\left\{
Z_{(i_1.\ldots,i_k)}^\bullet >z
\right\}
\right]
\\
&=&\mathbf{P}\left[
\left\{
Z_{{i_h,\overline{(i_1,\ldots,i_k)}}}^\bullet > Z_{(i_1,\ldots,i_k)}^\bullet >z \textnormal{ for all }i_h\in\{i_1,\ldots,i_k\}
\right\}
\right] \\
&=&
\mathbf{P}\left[
\left\{
E_{\Lambda_j}> Z_{(i_1,\ldots,i_k)}^\bullet >z \textnormal{ for all } j\in
\mathcal{RF}_{\overline{(i_1,\ldots,i_k)}}^l
\right\}
\right] \\
&=&
\alpha_{c,(i_1,\ldots,i_k)}\int_z^\infty (1+x)^{-\alpha_{c,\overline{(i_1,\ldots,i_k)}}}
(1+x)^{-\alpha_{c,(i_1,\ldots,i_k)}-1}dx
\\
&=&
\alpha_{c,(i_1,\ldots,i_k)}
\int_z^\infty (1+x)^{-\alpha_{c,i_1,\ldots,i_k}-1}dx
\end{eqnarray*}
and hence
\[
-\frac{d}{dz}\mathbf{P}[Z_{i_1}/\sigma_{i_1}=\cdots = Z_{i_k}/\sigma_{i_k}>z]=
 \alpha_{c,(i_1,\ldots,i_k)} (1+z)^{-\alpha_{c,i_1,\ldots,i_k}-1}.
\]

Further, as the integrand in (\ref{ThSingEq1}) is the d.d.f. of  a first to default r.v.
(Theorem \ref{max-th}), it is the  d.d.f. of a Lomax r.v. with random shape parameter
$\gamma_{c,i_1,\ldots,i_k}+\alpha_{c,i_1,\ldots,i_k}+K+1$, where $K$ is an integer valued
r.v. with p.m.f. \`a la (\ref{pk})  and scale parameter equal to one. Hence by Lemma
(\ref{FL2005}), we have that
\begin{eqnarray*}
\mathbf{P}[X_{i_1}/\sigma_{i_1}=\cdots=X_{i_k}/\sigma_{i_k}] \notag &=&
\alpha_{c,i_1,\ldots,i_k}\int_{\mathbf{R}_+} \sum_{k=0}^\infty p_k\left(
1+z
\right)^{-(\gamma_{c,i_1,\ldots,i_k}+\alpha_{c,i_1,\ldots,i_k}+k)-1}dz.
\end{eqnarray*}
The proof is completed by changing the order of summation and integration signs.
\end{proof}

\begin{proof}[Proof of Theorem \ref{decomp}]
Let $\mathcal{A}=\left \{(X_i,X_k)'\in\mathbf{R}_+^2:\frac{X_i}{\sigma_i}=\frac{X_k}{\sigma_k}\right \}$ for $1\leq i\neq k\leq n$
 denote the set on which the singular component of
the d.d.f. concentrates.  Then, for $a=\mathbf{P}[\ \mathcal{A}\ ]$ and $(x,y)'\in\mathbf{R}_+^2$,
\begin{eqnarray*}
\overline{F}(x,y) &=&a\mathbf{P}[\{(X_i>x) \cap ( X_k>y) |\ \mathcal{A}\}]+(1-a)\mathbf{P}[\{(X_i>x) \cap (X_k>y) |\ \mathcal{A}^c\}] \\
&=&\mathbf{P}[\{(X_i>x) \cap (X_k>y) \cap \mathcal{A}\}]+\mathbf{P}[\{(X_i>x) \cap (X_k>y) \cap \mathcal{A}^c\}] ,
\end{eqnarray*}
where
\begin{eqnarray*}
&&\mathbf{P}[\{(X_i>x) \cap (X_k>y) \cap \mathcal{A}\}] =
 \mathbf{P}\left[
 \left\{\left(\frac{X_i}{\sigma_i}>\frac{x}{\sigma_i}\right) \cap \left(\frac{X_k}{\sigma_k}>\frac{y}{\sigma_k}\right) \cap \mathcal{A}
 \right\}\right] \\
&=&
 \mathbf{P}\left[
 \frac{X_i}{\sigma_i}=\frac{X_k}{\sigma_k}>\frac{x}{\sigma_i}\bigvee \frac{y}{\sigma_k} \
 \right] \\
 &=&\alpha_{c,(i,k)}\int_{\frac{x}{\sigma_i} \vee \frac{y}{\sigma_k}}^\infty \sum_{h=0}^\infty
 (1+z)^{-\xi_{c,i,k}-h-1}p_hdz
\end{eqnarray*}
with the latter expression following employing the techniques used in the proof of Theorem  \ref{sim-default}. It then follows that
\[
\overline{F}_s(x,y)=\frac{\alpha_{c,(i,k)}}{a}\sum_{h=0}^\infty \frac{1}{\xi_{c,i,k}+h}\left(
1+\frac{x}{\sigma_i} \vee \frac{y}{\sigma_k}
\right)^{-(\xi_{c,i,h}+h)}p_h,
\]
which proves (\ref{scomp}). The form of the absolutely continuous component then follows from the latter expression, (\ref{ddfbiv})
and evoking Theorem \ref{sim-default}. This completes the proof.
\end{proof}

\begin{proof}[Proof of Theorem \ref{joint_m}]
Assume, without loss of generality, that $\sigma_i=\sigma_k=1$ for $1\leq i\neq k\leq n$. Then because of (\ref{ddfbiv}) the d.d.f. of $(X_i,X_k)'$ is given by
\begin{eqnarray*}
\overline{F}(x,y)&=&
\left(
1+x \vee y
\right)^{-\alpha_{c,(i,k)}}
\left(
1+x+y
\right)^{-\gamma_{c,(i,k)}}
\left(
1+x
\right)^{-\xi_{c,i,\overline{(i,k)}}}
\left(
1+y
\right)^{-\xi_{c,k,\overline{(i,k)}}} \notag \\
&=&
\left\{
\begin{array}{cc}
\left(
1+x+y
\right)^{-\gamma_{c,(i,k)}}
\left(
1+x
\right)^{-\xi_{c,i,\overline{(i,k)}}}
\left(
1+y
\right)^{-(\xi_{c,k,\overline{(i,k)}}+\alpha_{c,(i,k)})},
 & x\leq y \\
\left(
1+x+y
\right)^{-\gamma_{c,(i,k)}}
\left(
1+x
\right)^{-(\xi_{c,i,\overline{(i,k)}}+\alpha_{c,(i,k)})}
\left(
1+y
\right)^{-\xi_{c,k,\overline{(i,k)}}},& x>y
\end{array}
\right..
\end{eqnarray*}
Hence we readily have that
\begin{eqnarray*}
\mathbf{E}[X_iX_k]&=&\int_{\mathbf{R}_+^2} \overline{F}(x,y)dxdy \\
&=& \int_{\mathbf{R}_+}\int_x^\infty
\left(
1+x+y
\right)^{-\gamma_{c,(i,k)}}
\left(
1+x
\right)^{-\xi_{c,i,\overline{(i,k)}}}
\left(
1+y
\right)^{-(\xi_{c,k,\overline{(i,k)}}+\alpha_{c,(i,k)})}dydx
\\
&+&\int_{\mathbf{R}_+} \int_{0}^x
\left(
1+x+y
\right)^{-\gamma_{c,(i,k)}}
\left(
1+x
\right)^{-(\xi_{c,i,\overline{(i,k)}}+\alpha_{c,(i,k)})}
\left(
1+y
\right)^{-\xi_{c,k,\overline{(i,k)}}}dydx \\
&=& \int_{\mathbf{R}_+}\int_x^\infty
\left(
1+x+y
\right)^{-\gamma_{c,(i,k)}}
\left(
1+x
\right)^{-\xi_{c,i,\overline{(i,k)}}}
\left(
1+y
\right)^{-(\xi_{c,k,\overline{(i,k)}}+\alpha_{c,(i,k)})}dydx
\\
&+&\int_{\mathbf{R}_+} \int_{y}^\infty
\left(
1+x+y
\right)^{-\gamma_{c,(i,k)}}
\left(
1+x
\right)^{-(\xi_{c,i,\overline{(i,k)}}+\alpha_{c,(i,k)})}
\left(
1+y
\right)^{-\xi_{c,k,\overline{(i,k)}}}dxdy \\
&=&I_1(\boldsymbol{\xi})+I_2(\boldsymbol{\xi}).
\end{eqnarray*}

We further calculate $I_1(\boldsymbol{\xi})$, whereas the other integral can be tackled in a similar fashion. More specifically, by change
of variables and evoking Equation (3.197(1)) in Gradshteyn \& Ryzhik (2014), we have that
\begin{eqnarray*}
I_1(\boldsymbol{\xi)}&=&\frac{1}{\xi_{c,k}-1}\int_{\mathbf{R}_+}
(1+2x)^{-\gamma_{c,(i,k)}}
(1+x)^{-(\xi_{c,i,k}-\gamma_{c,(i,k)}-1)}{}_2F_1\left(
\gamma_{c,(i,k)},1;\xi_{c,k};\frac{x}{1+2x}
\right)
dx \\
&=&\frac{1}{2(\xi_{c,k}-1)}\int_0^1
(1-u/2)^{-(\xi_{c,i,k}-\gamma_{c,(i,k)}-1)}
(1-u)^{(\xi_{c,i,k}-3)}
{}_2F_1\left(
\gamma_{c,(i,k)},1;\xi_{c,k};\frac{u}{2}
\right)
du.
\end{eqnarray*}

Furthermore, note that as the ${}_2F_1$ hypergeometric function has the following integral representation for all
$\xi_{c,k}>1$ and $u\in\mathbf{R}$,
\[
\frac{1}{\xi_{c,k}-1}\ _2F_1\left(\gamma_{c,(i,k)},1;\xi_{c,k};u/2 \right)=\int_0^1 (1-v)^{\xi_{c,k}-2} \left(1-\frac{u}{2}v\right)^{-\gamma_{c,(i,k)}}dv
\]
(Equation 9.111 in loc. cit.) we obtain the following string of integrals
\begin{eqnarray*}
I_1(\boldsymbol{\xi})&=&
\frac{1}{2}\int_0^1 (1-v)^{\xi_{c,k}-2}\int_0^1 (1-u)^{\xi_{c,i,k}-3}(1-u/2)^{-(\xi_{c,i,k}-\gamma_{c,(i,k)}-1)}(1-uv/2)^{-\gamma_{c,(i,k)}}dudv
\\
&\overset{(1)}{=}&
\frac{1}{2(\xi_{c,i,k}-2)}\int_0^1
(1-v)^{\xi_{c,k}-2} F_1\left(
1,\gamma_{c,(i,k)},\xi_{c,i,k}-\gamma_{c,(i,k)}-1,\xi_{c,i,k}-1,\frac{v}{2},\frac{1}{2}
\right)
dv \\
&\overset{(2)}{=}&
\frac{1}{(\xi_{c,i,k}-2)}\int_0^1
(1-v)^{\xi_{c,k}-2} {}_2F_1\left(
1,\gamma_{c,(i,k)};\xi_{c,i,k}-1;v-1
\right)
dv \\
&=&
\frac{1}{(\xi_{c,i,k}-2)}\int_0^1
v^{\xi_{c,k}-2} {}_2F_1\left(
1,\gamma_{c,(i,k)};\xi_{c,i,k}-1;-v
\right)
dv \\
&\overset{(3)}{=}&
\frac{1}{(\xi_{c,i,k}-2)(\xi_{c,k}-1)}
{}_3F_2\left(\xi_{c,k}-1,1,\gamma_{c,(i,k)};\xi_{c,k},\xi_{c,i,k}-1;-1\right),
\end{eqnarray*}
where
$F_1$ is the bivariate hypergeometric function,
and  `$\overset{(1)}{=}$', `$\overset{(2)}{=}$' and `$\overset{(3)}{=}$' hold by
Equations (3.211), (9.182(1)) and (7.512(12)), respectively, in Gradshteyn \& Ryzhik (2014).
The expression for $I_2(\boldsymbol{\xi})$ is then by analogy
\[
I_2(\boldsymbol{\xi})=\frac{1}{(\xi_{c,i,k}-2)(\xi_{c,i}-1)}
{}_3F_2\left(\xi_{c,i}-1,1,\gamma_{c,(i,k)};\xi_{c,i},\xi_{c,i,k}-1;-1\right).
\]
We note in passing that the hypergeometric functions in $I_1(\boldsymbol{\xi})$ and  $I_2(\boldsymbol{\xi})$
converge for $\xi_{c,i,k}-\gamma_{c,(i,k)}>0$ and converge absolutely for  $\xi_{c,i,k}-\gamma_{c,(i,k)}>1$,
$1\leq i\neq k\leq n$.
This completes the proof.
\end{proof}

\begin{proof}[Proof of Corollary \ref{cov-su}]
First notice that according to Theorem \ref{joint_m}, we have that
\begin{eqnarray*}
&&\int_0^{\infty}\int_u^{\infty} (1+u)^{-c}(1+v)^{-b}(1+u+v)^{-a}dv du \\
&=&\frac{1}{(a+b+c-2)(a+b-1)}\ _3F_2(a+b-1,1,a;a+b,a+b+c-1;-1),
\end{eqnarray*}
where $a,b,c$ are all positive and such that $s_1:=a+b>2$ and $s_2:=a+c>2$. Then, for $s=a+b+c$,
\begin{eqnarray*}
&&\int_0^{\infty}\int_0^{\infty} (1+u)^{-c}(1+v)^{-b}(1+u+v)^{-a}dv du \\
&=&\int_0^{\infty}\int_v^{\infty} (1+u)^{-c}(1+v)^{-b}(1+u+v)^{-a}du dv+\int_0^{\infty}\int_u^{\infty} (1+u)^{-c}(1+v)^{-b}(1+u+v)^{-a}dv du\\
&=&\frac{1}{s-2}\left(\frac{1}{s_2-1}\ _3F_2(s_2-1,1,a;s_2,s-1;-1)+\frac{1}{s_1-1}\ _3F_2(s_1-1,1,a;s_1,s-1;-1)\right)\\
&\overset{(1)}{=}&\frac{1}{(s_1-1)(s_2-1)}\ _3F_2(a,1,1;s_1,s_2;1),
\end{eqnarray*}
where the latter equality is by Theorem 2.1 in Su \& Furman (2016), and the hypergeometric function converges absolutely
since $\xi_{c,i}>2$ and $\xi_{c,k}>2$ by assumption. This, along with setting $\alpha_{c,(i,k)}\equiv 0$
completes the proof.
\end{proof}

\begin{proof}[Proof of Corollary \ref{cov-cs}]
It is easy to check that, for $\gamma_{c,(i,k)}\equiv 0$, we have
${}_3F_2(a,b,0;c,d;z)\equiv 1$, for any real $a,b,c,d,z$ and also $\xi_{c,k}+\xi_{c,i}-\xi_{c,i,k}=\alpha_{c,(i,k)}$
for $1\leq i\neq k\leq n$. This completes the proof.
\end{proof}

\begin{proof}[Proof of Proposition \ref{survival-function}]
Let, for $(x,y)'\in\mathbf{R}_{+}^2$ and $b_1,b_2,b_3$ all real,
\begin{equation*}
\label{G-fun}
\overline{G}(x,y;b_1,b_2,b_3)=
\left(1+x\right)^{-b_1}
\left(1+y\right)^{-b_2}
\left(1+x+y\right)^{-b_3}.
\end{equation*}
Employing (\ref{ddfbiv}), we obtain that, for $y\in\mathbf{R}_{+}$,
\begin{eqnarray}
&&\mathbf{E}[X_i|X_k>y]=\int_{\mathbf{R}_+}\mathbf{P}[X_i>x|X_k>y]dx \notag \\
&=&\int_{\mathbf{R}_+}
\overline{G}
\left(
x/\sigma_i,y/\sigma_k;\ \xi_{c,i,\overline{(i,k)}},-\gamma_{c,(i,k)},\gamma_{c,(i,k)}
\right)dx \label{sfPropf1}
 \\
&-& \int_{\frac{\sigma_iy}{\sigma_k}}^\infty
\overline{G}
\left(
x/\sigma_i,y/\sigma_k;\ \xi_{c,i,\overline{(i,k)}},-\gamma_{c,(i,k)},\gamma_{c,(i,k)}
\right)dx \label{sfPropf2}\\
&+& \int_{\frac{\sigma_iy}{\sigma_k}}^\infty
\overline{G}
\left(
x/\sigma_i,y/\sigma_k;\ \xi_{c,i,\overline{(i,k)}}+\alpha_{c,(i,k)},-\xi_{c,(i,k)},\gamma_{c,(i,k)}
\right)dx.
\label{sfPropf3}
\end{eqnarray}
In order to compute (\ref{sfPropf1}), we use Expression (3.197(5)) in Gradshteyn \& Ryzhik (2014) and obtain that
\begin{eqnarray*}
\int_{\mathbf{R}_+} \overline{G}(x,y;b_1,-b_2,b_2)dx&=&
\int_{\mathbf{R}_+} (1+x)^{-b_1}\left(1+\frac{x}{1+y}\right)^{-b_2}dx \\
&=&\frac{1}{b_1+b_2-1}{}_2F_1\left(
b_2,1;b_1+b_2;\frac{y}{1+y}
\right),
\end{eqnarray*}
whereas to compute the other two integrals, i.e., (\ref{sfPropf2}) and (\ref{sfPropf3}), we use Expression (3.197(1)) in Gradshteyn \& Ryzhik (2014) and
have the following string of expressions
\begin{eqnarray*}
\int_z^\infty \overline{G}(x,y;b_1,-b_2,b_2)dx&=&\left(1+y\right)^{b_2}
\int_z^\infty \left(1+x\right)^{-b_1}
\left(
1+x+y
\right)^{-b_2}dx
\\
&=&\left(1+y\right)^{b_2}
\int_0^\infty \left(1+x+z\right)^{-b_1}
\left(
1+x+y+z
\right)^{-b_2}dx\\
&=&\frac{1}{b_1+b_2-1}\left(1+y\right)^{b_2}
(1+z)^{1-b_1}\left(
1+y+z
\right)^{-b_2} \\
&\times &{}_2F_1\left(
b_2,1;b_1+b_2;\frac{y}{1+y+z}
\right),
\end{eqnarray*}
where $z\in\mathbf{R}_+$ and $y\in\mathbf{R}_+$.
This, along with the formula for $\mathbf{E}[X]$ that has been derived in Theorem \ref{marginal-p}, completes the proof.
\end{proof}
\end{document}